\documentclass[11pt]{article}

\usepackage[breaklinks]{hyperref}  % comment out if you do not have it
\usepackage{graphicx}
\usepackage[active]{srcltx} % comment out if you do not have it
\usepackage{picins}
\usepackage{paralist}
\usepackage{ifpdf}

\providecommand{\TPDF}[2]{\texorpdfstring{#1}{#2}}

\newcommand{\IncGraphPageExt}[4]{
   \ifpdf
      \includegraphics[page=#3,#4]{{#1/#2}}
   \else
      \includegraphics[#4]{{#1/#2/#2_#3}}
   \fi
}

\newcommand{\aftermathA}{\par\vspace{-\baselineskip}}
\newenvironment{fake_proof}{{}}{\hfill{\hfill\rule{2mm}{2mm}}}
\usepackage{theorem}

\usepackage{amsmath}
\usepackage{xspace}
\usepackage{color}
\usepackage{euscript}
\usepackage{amssymb}

\newcommand{\Frechet}{Fr\'{e}chet\xspace}%

\newcommand{\distFr}[2]{d_{\EuScript{F}}\pth{#1, #2}}
\newcommand{\distCmd}[1]{\left\| {#1} \right\|}
\newcommand{\distX}[2]{\distCmd{#1 - #2}}

\providecommand{\nfrac}[2]{#1/#2}
\newcommand{\curveA}{\pi}
\newcommand{\subcurveA}{\widehat{\pi}}
\newcommand{\curveAs}{\pi'}
\newcommand{\curveB}{\sigma}
\newcommand{\subcurveB}{\widehat{\sigma}}
\newcommand{\curveBs}{\sigma'}
\newcommand{\curveC}{\tau}
\newcommand{\curveD}{\eta}

\newcommand{\Holder}{H\"older\xspace}

\newcommand{\segA}{\mathsf{u}}
\newcommand{\segB}{\mathsf{v}}

\newcommand{\pntA}{p}
\newcommand{\pntB}{q}
\newcommand{\pntC}{s}
\newcommand{\pntD}{t}
\newcommand{\pntE}{r}

\newcommand{\BallC}{B}
\newcommand{\BallX}[1]{B\pth{#1}}
\newcommand{\SphereX}[1]{S\pth{#1}}

\newcommand{\lenX}[1]{\left\|{#1} \right\|}

\newcommand{\FDleqC}{\EuScript{R}}
\newcommand{\FDleq}[1]{\FDleqC_{\leq #1}}

\newcommand{\ShowComments}[1]{}

\newcommand{\FullFDleqC}{D}
\newcommand{\FullFDleq}[1]{\FullFDleqC_{\leq #1}}
\newcommand{\NleqC}{N}
\newcommand{\Nleq}[1]{\NleqC_{\leq #1}}

\newcommand{\Graph}{G}
\newcommand{\Cone}{\mathcal{V}}
\newcommand{\Family}{\mathcal{F}}
\providecommand{\pth}[2][\!]{#1\left({#2}\right)}
\providecommand{\brc}[1]{\left\{ {#1} \right\}}
\providecommand{\sep}[1]{\,\left|\, {#1} \MakeBig\right.}
\providecommand{\MakeBig}{\rule[-.2cm]{0cm}{0.4cm}}

\newcommand{\diameterX}[1]{\mathrm{d{}i{}am}\pth{#1}}
\newcommand{\constA}{c_1}
\newcommand{\constB}{c_2}
\newcommand{\constC}{c_3}

\newcommand{\cSimpB}{c}

\newcommand{\cSimpBVal}{1/3}

\newcommand{\PntSet}{{\mathsf{P}}}
\newcommand{\PntSetA}{{\mathsf{Q}}}

\newcommand{\sRadius}{\mu}
\newcommand{\Radius}{\mathsf{r}}
\newcommand{\Grid}{\EuScript{G}}
\newcommand{\CellXY}[2]{C_{#1,#2}}

\providecommand{\MakeBig}{\rule[-.2cm]{0cm}{0.4cm}}
\providecommand{\MakeSBig}{\rule[0.0cm]{0.0cm}{0.35cm}} % really small

\newcommand{\IVCellXY}[2]{I^v_{#1,#2}}
\newcommand{\IHCellXY}[2]{I^h_{#1,#2}}
\newcommand{\RVCellXY}[2]{R^v_{#1,#2}}
\newcommand{\RHCellXY}[2]{R^h_{#1,#2}}

\newcommand{\Cell}{\mathsf{R}}

\newcommand{\CubeSet}{\EuScript{K}}
\newcommand{\Cube}{C}

\providecommand{\cardin}[1]{\lvert#1\rvert}
\newcommand{\VertexSet}[1]{{V}\pth{#1}}

\newcommand{\KSet}{\EuScript{K}}
\newcommand{\veEvents}{Z}
\newcommand{\sEvents}{Z}

\providecommand{\eps}{{\varepsilon}}%

\providecommand{\ceil}[1]{\left\lceil {#1} \right\rceil}
\providecommand{\floor}[1]{\left\lfloor {#1} \right\rfloor}

\newcommand{\density}{\phi}
\renewcommand{\th}{t{}h\xspace}
\newcommand{\nVertices}[1]{\mathsf{n}_{#1}}
\newcommand{\nVCell}[2]{\cardin{{#1} \sqcap{#2}}\,}

\newcommand{\RX}[1]{\begin{minipage}{0.66cm}
       \vspace{-0.6cm}
       $\Longrightarrow${#1}
       \vspace{0.6cm}
       \end{minipage}}

\newcommand{\Hippodrome}{\mathcal{H}}

\newcommand{\SimpComplexityC}{\mathsf{N}}
\newcommand{\SimpComplexity}[3]{\SimpComplexityC\pth{#1, #2, #3}}

\newlength{\savedparindent}
\newcommand{\SaveIndent}{\setlength{\savedparindent}{\parindent}}
\newcommand{\RestoreIndent}{\setlength{\parindent}{\savedparindent}}

\newlength{\ppicX}
\newlength{\ppicY}
\newcommand{\ParaWPic}[2]{
   \settowidth{\ppicX}{#1}  %% save width
   \setlength{\ppicY}{\linewidth}
   \addtolength{\ppicY}{-\ppicX}
   \addtolength{\ppicY}{-0.30cm}

   \vspace{-0.2cm}
   \SaveIndent
   \noindent\begin{minipage}[t]{\ppicY}%
       \vspace{0pt}
       \RestoreIndent
       #2
   \end{minipage}%
   \hfill %
   \begin{minipage}[t]{\ppicX}
       \vspace{0pt}
       #1
       \vfill
   \end{minipage}
   \smallskip
}

\newcommand{\atgen}{\symbol{'100}}
\newcommand{\SarielThanks}[1]{\thanks{Department of Computer
      Science; 
      University of Illinois; 
      201 N. Goodwin Avenue;
      Urbana, IL, 61801, USA;
      {\tt sariel\atgen{}uiuc.edu}; {\tt
         \url{http://www.uiuc.edu/\string~sariel/}.} #1}}

\definecolor{blue25}{rgb}{0,0,0.45}%
\newcommand{\emphic}[2]{%
   \textcolor{blue25}{%
      \textbf{\emph{#1}}}%
   \index{#2}}

\newcommand{\emphi}[1]{\emphic{#1}{#1}}

\newtheorem{theorem}{Theorem}[section]

\newtheorem{lemma}[theorem]{Lemma}

\newtheorem{claim}[theorem]{Claim}

\newcommand{\pbrc}[1]{\left[ {#1} \MakeBig \right]}
\renewcommand{\Re}{{\rm I\!\hspace{-0.025em} R}}

{\theorembodyfont{\rm} }
{\theorembodyfont{\rm} \newtheorem{defn}[theorem]{Definition}}
{\theorembodyfont{\rm} \newtheorem{algorithm}[theorem]{Algorithm}}
{\theorembodyfont{\rm} }
{\theorembodyfont{\rm} \newtheorem{remark}[theorem]{Remark}}
{\theorembodyfont{\rm} \newtheorem{observation}[theorem]{Observation}}

\DefineNamedColor{named}{RedViolet} {cmyk}{0.07,0.90,0,0.34}

\newcommand{\AlgorithmI}[1]{{\textcolor[named]{RedViolet}{\texttt{\bf{#1}}}}}
\newcommand{\Algorithm}[1]{{\AlgorithmI{#1}\index{#1@{\AlgorithmI{#1}}}}}

\newcommand{\seclab}[1]{\label{sec:#1}}
\newcommand{\secref}[1]{Section~\ref{sec:#1}}

\newcommand{\clmlab}[1]{\label{claim:#1}}
\newcommand{\clmref}[1]{Claim~\ref{claim:#1}}

\newcommand{\obslab}[1]{\label{observation:#1}}
\newcommand{\obsref}[1]{Observation~\ref{observation:#1}}

\newcommand{\defref}[1]{Definition~\ref{def:#1}}
\providecommand{\deflab}[1]{\label{def:#1}}

\newcommand{\algref}[1]{Algorithm~\ref{alg:#1}}
\providecommand{\alglab}[1]{\label{alg:#1}}

\newcommand{\lemlab}[1]{\label{lemma:#1}}
\newcommand{\lemref}[1]{Lemma~\ref{lemma:#1}}

\newcommand{\apndlab}[1]{\label{apnd:#1}}
\newcommand{\apndref}[1]{Appendix~\ref{apnd:#1}}

\newcommand{\thmref}[1]{Theorem~\ref{theo:#1}}
\newcommand{\thmlab}[1]{{\label{theo:#1}}}

\newcommand{\remlab}[1]{\label{rem:#1}}
\newcommand{\remref}[1]{Remark~\ref{rem:#1}}

\newenvironment{proof}{\trivlist\item[]\emph{Proof}:}%
                  {\unskip\nobreak\hskip 1em plus 1fil\nobreak%
                           \rule{2mm}{2mm}%$\Box$
                           \parfillskip=0pt%
                           \endtrivlist}
\providecommand{\ds}{\displaystyle}

\newcommand{\simpX}[1]{\mathrm{s{i}m{p}l}\pth{#1}}
\newcommand{\deciderFr}{\Algorithm{decider}\xspace}
\newcommand{\intervalFr}{\Algorithm{search{}Interval}\xspace}
\newcommand{\intervalExactFr}{\Algorithm{search{}Interval{}No{}Si{}mp}\xspace}
\newcommand{\approxFr}      {\Algorithm{a{}p{}r{}x{}\Frechet{}I}\xspace}
\newcommand{\approxFrDirect}{\Algorithm{a{}p{}r{}x{}\Frechet{}No{}Si{}mp}\xspace}
\newcommand{\approxBinarySearch}{\Algorithm%
   {search{}Events}\xspace}
\newcommand{\approxDistances}{\Algorithm{approx{}Distances}\xspace}

\newcommand{\etal}{\textit{e{}t~a{}l.}\xspace}
\newcommand{\widthX}[2]{{\mathrm{width}}_{#2}\pth{#1}}
\newcommand{\R}{\mathbb{R}}

\newcommand{\XSaysExt}[2]{{
      ~\\
      \fbox{\begin{minipage}{0.99\linewidth}
             {$\rule[-0.12cm]{0.2in}{0.5cm}$\fbox{\tt
                   #1:}} #2
         \end{minipage}}
      \marginpar{#1}}}

\newcommand{\SarielX}[1]{\XSaysExt{Sariel}{#1}}

\newcommand{\AnneX}[1]{{\XSaysExt{Anne}{#1}}}

\newcommand{\Links}{\mathsf{d}}

\newcommand{\Interval}{\mathcal{I}}
\newcommand{\wInterval}{\mathcal{J}}

\newcommand{\AlphaBetaCovered}{$(\alpha,\beta)$-covered\xspace}

\newcommand{\remove}[1]{}

\newcommand{\si}[1]{#1}

\newcommand{\BFS}{\textsc{bfs}\xspace}
\newcommand{\Pairwise}[1]{\EuScript{D}(#1)}
\newcommand{\PW}{\mathcal{W}}

\newcommand{\figlab}[1]{\label{fig:#1}}
\newcommand{\figref}[1]{Figure~\ref{fig:#1}}

\newcommand{\loLimit}{h^-}
\newcommand{\hiLimit}{h^+}

\newcommand{\ratioX}[2]{[ #1/#2 ]}

\newcommand{\approxDist}{\delta}

\newcommand{\tabref}[1]{Table~\ref{tab:#1}}
\newcommand{\tablab}[1]{\label{tab:#1}}

\newcommand{\VolumeX}[1]{\mathrm{vol}\pth{#1}}

\newcommand{\resemblance}{relative free space complexity\xspace}
\newcommand{\resemblanceII}{relative complexity of the free
   space\xspace} \newcommand{\Resemblance}{Relative Free Space
   Complexity\xspace}
\newcommand{\resemblanceX}[1]{\ensuremath{#1}-relative free space
   complexity\xspace}

\newcommand{\relevant}{reachable\xspace}
\newcommand{\Relevant}{Reachable\xspace}
\newcommand{\cDim}{c_{d}}

\begin{document}

\title{Approximating the \Frechet Distance for Realistic Curves in
   Near Linear Time%
   \thanks{The latest full version of this paper is available online
      \cite{dhw-afdrc-10}.}}

% \InSubmitVer{\thispagestyle{empty} }
\author{%
   Anne Driemel%
   \thanks{Department of Information and Computing Sciences; Utrecht
      University; The Netherlands;
      \texttt{anne}\hspace{0cm}\texttt{\atgen{}cs.uu.nl}.  This work
      has been supported by the Netherlands Organisation for
      Scientific Research (NWO) under \si{RIMGA} (Realistic Input Models
      for Geographic Applications).} %
   \and%
   Sariel Har-Peled%
   \SarielThanks{Work on this paper was partially supported by a NSF
      AF award CCF-0915984.}%
   \and %
   Carola Wenk%
   \thanks{Department of Computer Science; University of Texas at San
      Antonio; One UTSA Circle; San Antonio, TX 78249-0667, USA; {\tt
         carola\atgen{}cs.utsa.edu}.  This work has been supported by
      NSF CAREER award CCF-0643597.}  }

\date{\today}

\maketitle

\begin{abstract}
    We present a simple and practical $(1+\eps)$-approximation
    algorithm for the \Frechet distance between two polygonal curves
    in $\Re^d$.  To analyze this algorithm we introduce a new
    realistic family of curves, $c$-packed curves, that is closed
    under simplification. We believe the notion of $c$-packed curves
    to be of independent interest.  We show that our algorithm has
    near linear running time for $c$-packed polygonal curves, and
    similar results for other input models, such as low density
    polygonal curves.
\end{abstract}

\section{Introduction}

Comparing geometric shapes is a task that arises in a wide arena of
applications. The \Frechet{} distance and its variants have been used,
to this end, to compare curves in applications such as dynamic
time-warping \cite{kp-sudtw-99}, speech recognition
\cite{khmtc-pgbhp-98}, signature and handwriting recognition
\cite{mp-cdtw-99, skb-fdba-07}, matching of time series in databases
\cite{kks-osmut-05}, as well as geographic applications, such as
map-matching of vehicle tracking data
\cite{bpsw-mmvtd-05,wsp-anmms-06}, and moving objects analysis
\cite{bbg-dsfm-08, bbgll-dcpcs-08}.

\parpic[r]{\includegraphics{figs/\si{hausdorff_sucks_II}}}

Informally, the \Frechet distance between two curves is the maximum
distance a point on the first curve has to travel as this curve is
being continuously deformed into the second curve, see
\secref{f:distance} for the formal definition.  Unlike the Hausdorff
distance, which is solely based on nearest neighbor distances between
points on the curves, the \Frechet distance requires continuous and
order-preserving assignments of points and hence is better suited for
comparing curves with respect to their intrinsic structure.
% See the figure to the right for an example of two dissimilar curves
% that have a small Hausdorff distance.

\parpic[l]{\includegraphics{figs/\si{hausdorff_sucks}}}
The \Frechet distance between two curves might be arbitrarily larger
than their Hausdorff distance, as demonstrated by the figure on the
left, and as this example shows, it seems to be a more natural measure
of similarity between curves.

\paragraph{Previous results.}
For two polygonal curves of total complexity $n$ in the plane, their
\Frechet distance can be computed in $O(n^2\log n)$ time
\cite{ag-cfdbt-95}, and their Hausdorff distance can be computed in
$O(n\log n)$ time \cite{a-cgcs-09}. It has been an open problem to
find a subquadratic algorithm for computing the \Frechet distance for
two curves. For the problem of deciding whether the \Frechet distance
between two curves is smaller or equal a given value a lower bound of
$\Omega(n \log n)$ was given by \cite{bbkrw-wtd-07}. Recently, Alt
\cite{a-cgcs-09} conjectured that the decision problem may be
3SUM-hard.  The only subquadratic algorithms known are for quite
restricted classes of curves such as for closed convex curves and for
$\kappa$-bounded curves \cite{akw-cdmpc-04}. For a curve to be
$\kappa$-bounded means, roughly, that for any two points on the curve
the portion of the curve in between them cannot be further away from
either point than $\kappa/2$ times the distance between the two
points. For closed convex curves the \Frechet distance equals the
Hausdorff distance and for $\kappa$-bounded curves the \Frechet
distance is at most $(1+\kappa)$ times the Hausdorff distance, and
hence the $O(n\log n)$ algorithm for the Hausdorff distance applies.

Aronov \etal \cite{ahkww-fdcr-06} provided a near linear time
$(1+\eps)$-approximation algorithm for the \emph{discrete} \Frechet
distance, which only considers distances between vertices of the
curves. Their algorithm works for \emph{backbone curves}, which are
used to model protein backbones in molecular biology. Backbone curves
are required to have, roughly, unit edge length and a minimal distance
between any pair of vertices.  They use curve simplification to speed
up their algorithm.  Agarwal \etal \cite{ahmw-nltaa-05} studied fast
simplification that preserves the \Frechet distance.

\paragraph{The input model.}
We introduce a new class of curves, called $c$-packed curves, for
which we can approximate the \Frechet{} distance quickly, given that
the constant $c$ is small.  Intuitively, the constant $c$ measures how
``\emph{unrealistic}'' the input is.  We compare this new input model
to previous models such as \emph{fatness} and \emph{low density}, as
well as \emph{$\kappa$-boundedness}.  These so-called \emph{realistic
   input models} are commonly used for the analysis of problems where
the worst case complexity is dominated by degenerate or contrived
configurations which are highly unlikely to occur in practice, see
\cite{bksv-rimga-02} for an overview.

A curve $\curveA$ is \emphi{$c$-packed} if the total length of
$\curveA$ inside any ball is bounded by $c$ times the radius of the
ball.  A $\kappa$-bounded curve might have arbitrary length while
maintaining a finite diameter, and as such may not be $c$-packed, see
\secref{kappa:bounded:curves}. But unlike $\kappa$-bounded curves, the
\Frechet distance between two $c$-packed curves might be arbitrarily
larger than their Hausdorff distance.  Indeed, $c$-packed curves are
considerably more general and a more natural family of curves.  For
example, a $c$-packed curve might self cross and revisit the same
location several times, and the class of $c$-packed curves is closed
under concatenation, none of which is true for $\kappa$-bounded
curves.  Intuitively, $c$-packed curves behave reasonably in any
resolution.

\parpic[r]{\includegraphics{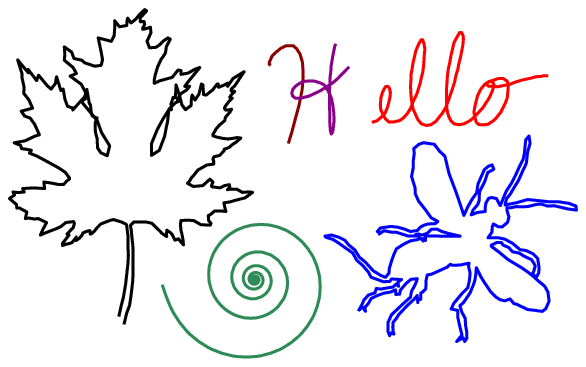}}
See the figure on the right for a few examples of $c$-packed curves.
The boundary of convex polygons, algebraic curves of bounded maximum
degree, the boundary of \AlphaBetaCovered shapes \cite{e-cuabc-05},
and the boundary of $\gamma$-fat shapes \cite{d-ibucfo-08} are all
$c$-packed. Indeed, the boundaries of \AlphaBetaCovered shapes and
$\gamma$-fat shapes are assumed to be formed by a constant number of
algebraic curves of bounded maximum degree.  If one removes the
requirement that a $\gamma$-fat curve be of bounded descriptive
complexity, then also fractal curves, like the Koch's snowflake, which
can have infinite length within a bounded area, can be fat
\cite{bcd-npfo-11}. Naturally, these curves cannot be $c$-packed.
Interestingly, one can show that \AlphaBetaCovered polygons are
$c$-packed even if they have unbounded complexity, see
\apndref{a:b:covered:c:pack} and also the result of Bose \etal
\cite{bcd-npfo-11}.

It is easy to verify that $c$-packed curves are also low density
\cite{bksv-rimga-02}, but a low density curve might not be $c$-packed,
for any bounded $c$, see \secref{low:density}. However, the class of
$c$-packed curves is closed under simplification, see
\lemref{6:c:packed}, and this is not true for low density curves.

\newcommand{\UBig}{\rule[-.2cm]{0cm}{0.7cm}}
\newcommand{\UBigger}{\rule[-.2cm]{0cm}{1cm}}
\begin{table}[t]
    \centerline{
       \begin{tabular}{|l|l|l|}
           \hline
           Curves type & Running time  & See\\
           \hline
           \hline     
           % Any &
           % $\UBig \ds O( \SimpComplexity{\eps}{\curveA}{\curveB} +
           % \SimpComplexity{1}{\curveA}{\curveB} \log n)$
           % & \thmref{main} \\ [0.1cm]
           \hline
           $c$-packed & $\UBig \ds O( c n /\eps + cn \log n)$ &
           \thmref{main:c:packed}\\
           \hline
           $\kappa$-straight & 
           Same as $2\kappa$-packed&
           \lemref{kappa:straight}\\
           \hline
           $\kappa$-bounded & 
           $\UBig \ds O\pth{ (\kappa
              /\eps)^d n + \kappa^d n \log n}$ & 
           \thmref{main:kappa:bounded}\\
           \hline
           $O(1)$-low density &  
           \begin{minipage}{0.38\linewidth}
               \smallskip $\ds O \pth{ \frac{n^{2(d-1)/d}}{\eps^2} +
                  n^{2(d-1)/d} \log n }$ \smallskip
           \end{minipage}
           &
           \thmref{main:low:density}    \\
           \hline
           $c$-packed \& closed &
           $\UBig \ds O\pth{ c^2 n/\eps^2 + 
              c^2 n  \log n}$ &
           \thmref{main:closed}\\        
           \hline\hline
       \end{tabular}
    }
    
    \caption{Summary of new results for computing 
       a $(1+\eps)$-approximation to the \Frechet
       distance between two curves $\curveA$ and $\curveB$ with $n$
       vertices in $\Re^d$. 
    }%
       
    \tablab{o:g:i:b:o:g:i}
\end{table}

\paragraph{Our results.}
We present a new algorithm for computing a $(1+\eps)$-approximation of
the \Frechet distance for polygonal curves in $\Re^d$.  Underlying the
algorithm are several new insights.  First, we use the idea of curve
simplification to reduce the complexity of the free space diagram, as
this simplification results in a contraction of the corresponding rows
or columns in the free space diagram. We introduce the notion of
\emph{\resemblance} in \defref{simplification:complexity} to capture
the complexity of the free space diagram of two curves, which are
simplified to the appropriate resolution.  Surprisingly, without
simplification, almost any two curves from natural families of curves
can have a free space diagram for the value realizing the \Frechet
distance that has quadratic complexity (even in the plane).  Secondly,
we present an efficient construction algorithm for this reduced size
free space diagram that enables us to solve the decision problem in
linear time in the \resemblance of the curves.  Thirdly, we prove that
monotonicity events are sufficiently close to vertex-edge events or an
approximate distance between two vertices of the curves. Therefore,
the search for the \Frechet distance can be done efficiently without
using parametric search or random sampling, by using approximate
distance selection.  Carefully combining these insights yields the new
algorithm, which has running time near linear in the \resemblance of
the input curves.

In the second part of the paper, we analyze the \resemblance for
various families of curves.  We prove that $c$-packed curves have
linear \resemblance for fixed $c$ and $\eps$.  We next prove a
subquadratic bound on the \resemblanceII of low density curves.  This
relies on a new packing lemma showing that, if the simplification of a
low density curve is long inside a relatively small area, then the
original curve must contain many vertices in the vicinity of this
region. We also prove that the \resemblance of $\kappa$-bounded curves
is linear for a fixed $\kappa$, which leads to an improvement of the
result by Alt \etal \cite{akw-cdmpc-04}.

These bounds imply that the approximation algorithm provides fast
approximation for the \Frechet distance for all these types of
curves. We also show how to adapt our algorithm to handle closed
curves.  The new results are summarized in \tabref{o:g:i:b:o:g:i}.

\paragraph{Organization.}
In \secref{prelims}, we provide some background on the \Frechet
distance and the notion of the free space diagram.  In
\secref{algorithm}, we describe the approximation algorithm that uses
simplification.  To this end, we show in \secref{relevant} that it
suffices to only compute the \relevant parts of the free space diagram
and in \secref{decision:procedure} we present a fuzzy decider
procedure and show how it can be used to make exact decisions during a
binary search for the \Frechet distance.  In \secref{search}, we deal
with the different subroutines used in the search for the \Frechet
distance and in \secref{the:algorithm} we give the resulting general
algorithm and analyze its correctness and running time, which is near
linear in the \resemblance.  In \secref{resemblance}, we bound the
\resemblance of various families of curves.  In particular, in
\secref{c:packed:curves}, we introduce the notion of $c$-packed
curves, and study their behavior under simplification.  In
\secref{kappa:bounded:curves}, we bound the \resemblance of
$\kappa$-bounded curves, and in \secref{low:density} we handle low
density curves.  In \secref{closed:curves}, we extend the algorithm to
closed curves.  We conclude with discussion in \secref{conclusions}.

\section{Preliminaries}
\seclab{prelims}

\subsection{Notations and Definitions}

Let $\curveA$ be a curve in $\R^d$; that is, a continuous mapping from
$[0,1]$ to $\Re^d$. In the following, we will identify $\curveA$ with
its range $\curveA\pth{[0,1]} \subseteq \R^d$ if it is clear from the
context. The curve $\curveA$ is \emphi{closed} if
$\curveA(0)=\curveA(1)$.  We use $\lenX{\cdot}$ to denote the
Euclidean distance as well as the length of a curve.  For a polygonal
curve $\curveA$, let $\VertexSet{\curveA}$ denote the set of vertices
of $\curveA$. For two points $\pntA$ and $\pntB$ on a curve $\curveA$,
let $\curveA[\pntA,\pntB]$ denote the portion of the curve between the
two points.

We denote with $\BallX{\pntA, r}$ the ball of radius $r$ centered at
$\pntA$, and $\SphereX{\pntA, r}$ denotes the corresponding sphere.
Given a set of numbers $U \subseteq \Re$, an \emphi{atomic interval}
of $U$ is a (possibly infinite) maximal interval on the real line that
does not contain any point of $U$ in its interior.  Let
$\Pairwise{\PntSet}$ be the set of all pairwise distances of points in
$\PntSet$.

\subsection{\Frechet Distance and the Free Space Diagram}
\seclab{f:distance}

A \emphi{reparameterization} is a bijective and continuous function
$f:[0,1]\rightarrow [0,1]$. It is \emphi{orientation-preserving} if
$f(0)=0$ and $f(1)=1$. Given two reparameterizations $f$ and $g$ for
two curves $\curveA$ and $\curveB$, respectively, define their
\emphi{width} as
\begin{align*}
    % $
    \widthX{\curveA,\curveB}{f,g} = \max_{s \in [0,1]}
    \distX{\curveA(f(s))}{\curveB(g(s))}.
    % $
\end{align*}
This can be interpreted as the maximum length of a leash one needs to
walk a dog, where the dog walks monotonically along $\curveA$
according to $f$, while the handler walks monotonically along
$\curveB$ according to $g$. In this analogy, the \Frechet distance is
the shortest possible leash admitting such a walk.

Formally, given two curves $\curveA$ and $\curveB$ in $\Re^d$, the
\emphi{\Frechet distance} between them is
\begin{align*}
    \distFr{\curveA}{\curveB} = \inf_{%
       \substack{f:[0,1] \rightarrow [0,1]\\
          g:[0,1]\rightarrow [0,1]}} \widthX{\curveA,\curveB}{f,g},
\end{align*}
where $f$ and $g$ are orientation-preserving reparameterizations of
the curves $\curveA$ and $\curveB$, respectively.  The \Frechet
distance complies with the triangle inequality; that is, for any three
curves $\curveA, \curveB$ and $\curveC$ we have that
$\distFr{\curveA}{\curveC} \leq \distFr{\curveA}{\curveB}
+\distFr{\curveB}{\curveC}$.

Let $\curveA$, $\curveB$ be curves and $\delta>0$ a parameter, the
\emphi{free space} of $\curveA$ and $\curveB$ of radius $\delta$ is
defined as
\begin{align*}
    \FullFDleq{\delta}(\curveA,\curveB) = \brc{ (s,t) \in [0,1]^2
       \sep{ \distX{\curveA(s)}{\curveB(t)} \leq \delta}}.
\end{align*}
\ We are interested only in polygonal curves.  Then the square
$[0,1]^2$ can be broken into a (not necessarily uniform) grid called
the \emphi{free space diagram}, where a vertical line corresponds to a
vertex of $\curveA$ and a horizontal line corresponds to a vertex of
$\curveB$. Every two segments of $\curveA$ and $\curveB$ define a
\emphi{free space cell} in this grid.  In particular, let
$\CellXY{i}{j} = \CellXY{i}{j}\pth{ \curveA, \curveB}$ denote the free
space cell that corresponds to the $i$\th edge of $\curveA$ and the
$j$\th edge of $\curveB$. The cell $\CellXY{i}{j}$ is located in the
$i$\th column and $j$\th row of this grid.

\parpic[r]{\includegraphics[scale=0.8]{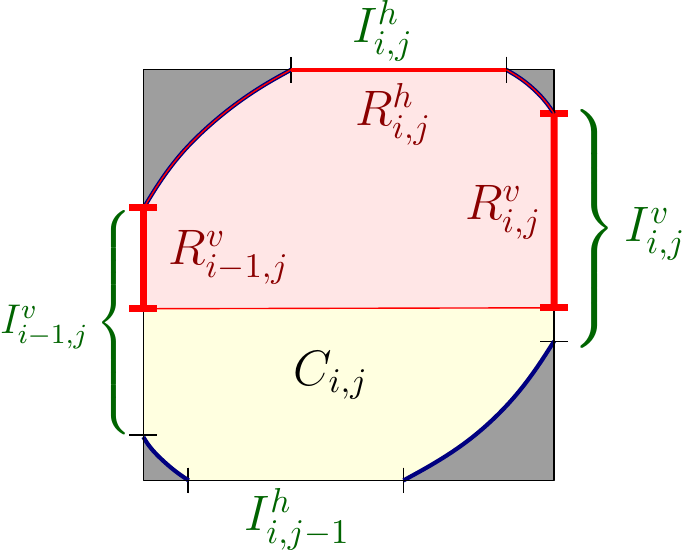}}

It is known that the free space, for a fixed $\delta$, inside such a
cell $\CellXY{i}{j}$ (i.e., $\FullFDleq{\delta}(\curveA, \curveB) \cap
\CellXY{i}{j}$) is the clipping of an affine transformation of a disk
to the cell \cite{ag-cfdbt-95}, see the figure to the right; as such,
it is convex and of constant complexity.  Let $\IHCellXY{i}{j}$ denote
the horizontal \emphi{free space interval} at the top boundary of
$\CellXY{i}{j}$, and $\IVCellXY{i}{j}$ denote the vertical free space
interval at the right boundary.

The \Frechet distance between $\curveA$ and $\curveB$ is at most
$\delta$ if and only if there is an $(x,y)$-monotone path in the free
space diagram between $(0,0)$ and $(1,1)$ that is fully contained in
$\FullFDleq{\delta}(\curveA,\curveB)$.  Let the \emphi{reachability
   intervals} $\RHCellXY{i}{j}\subseteq \IHCellXY{i}{j}$ and
$\RVCellXY{i}{j}\subseteq\IVCellXY{i}{j}$ consist of the points
$(x,y)$ on the boundary that are reachable by a monotone path from
$(0,0)$ to $(x,y)$.

Such a path to $(1,1)$ can be computed, if it exists, in $O(n^2)$ time
by dynamic programming, where $n$ is the total complexity of the two
polygonal curves $\curveA$ and $\curveB$, see \cite{ag-cfdbt-95}.

\subsubsection{Free Space Events}
\seclab{f:s:events}

To compute the \Frechet distance consider increasing $\delta$ from $0$
to $\infty$. As $\delta$ increases, structural changes to the free
space happen. We are interested in the radii (i.e., the value of
$\delta$) of these events.

\parpic[l]{\includegraphics{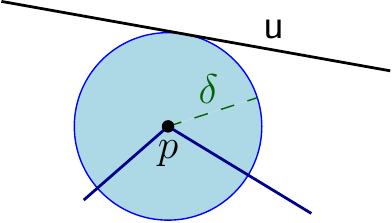}}
Consider a segment $\segA$ of $\curveA$ and a vertex $\pntA$ of
$\curveB$, a \emphi{vertex-edge event} corresponds to the minimum
value $\delta$ such that $\segA$ is tangent to $\BallX{\pntA,
   \delta}$.  In the free space diagram, this corresponds to the event
that a free space interval that consists of only one point was just
created.  The line supporting this boundary edge corresponds to the
vertex, and the other dimension corresponds to the edge.  Naturally,
the event could happen at a vertex of $\segA$.

% 
% a cell boundary has non-empty intersection with $\FullFDleqC =
% \FullFDleq{\delta}(\curveA, \curveB)$.

% \vspace{-0.5cm}

% \parpic[r]{\includegraphics[scale=0.9]{figs/monotone_event}}
% \parpic[r]{\includegraphics[scale=1]{figs/f_s_diagram_m_e_curves}}
% \vspace{0.5cm}

The second type of event, a \emphi{monotonicity event}, corresponds to
a value $\delta$ for which a monotone subpath inside
$\FullFDleq{\delta}$ becomes feasible, see \figref{monotonicity}.
Geometrically, this corresponds to two vertices $\pntA$ and $\pntB$ on
one curve and a directed segment $\segA$ on the other curve such that:
(1) $\segA$ passes through the intersection $\SphereX{\pntA,\delta}
\cap \SphereX{\pntB, \delta}$, and (2) $\segA$ intersects
$\BallX{\pntB, \delta}$ first and $\BallX{\pntA,\delta}$ second, where
$\pntA$ comes before $\pntB$ in the order along the curve $\curveA$.

Other values of $\delta$ that would be relevant to our algorithm are
the distances between any pair of points of $\VertexSet{\curveA} \cup
\VertexSet{\curveB}$. Technically, apart from the two single events
that the endpoints of the curves are being matched to each other,
these \emphi{vertex-vertex} events are vertex-edge events when they
are relevant, but they will be handled naturally by our algorithm.

\begin{figure}
    \center
    \IncGraphPageExt{figs}{f_s_diagram_m_e_curves}{1}{scale=1.2}%
    ~ ~
    \includegraphics[scale=1]{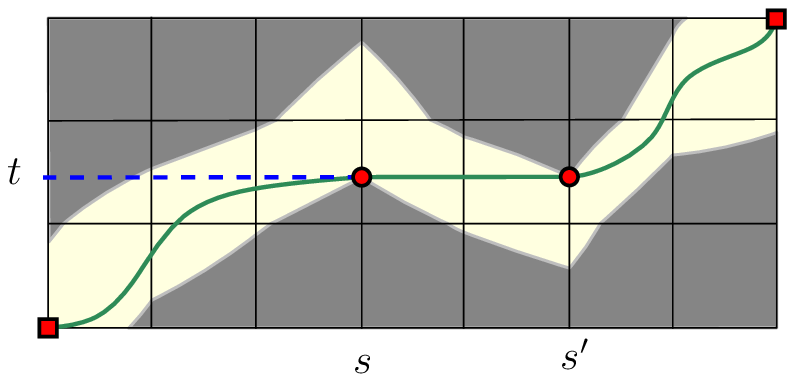}
    \caption{%
       Two curves $\curveA$ and $\curveB$ and their free space diagram
       $\FullFDleq{\delta}(\curveA,\curveB)$, where $\pntA=\curveA(s),
       \pntB=\curveA(s')$ and $\pntE=\curveB(t)$. Here, $\delta$ is
       the minimal free space parameter, such that a monotone path
       exists, i.e., in this example $\distFr{\curveA}{\curveB}$
       coincides with a monotonicity event. }
    \figlab{monotonicity}
\end{figure}

% 
% \begin{remark}
%     It might happen that two long edges intersect in their middle
%     (or in higher dimension pass close to each other), and thus
%     contribute an isolated connected component to
%     $\FullFDleq{\delta}(\curveAs,\curveBs)$.  Such a connected
%     component is a convex set lying completely in the interior of
%     the grid cell of the two segments. Since it is not reachable by
%     a monotone path in the diagram from $(0,0)$, we can just ignore
%     such \emphi{edge-edge} events. Such a connected component would
%     grow as $\delta$ increases until it hits the boundary of the
%     grid cell. At this point, a vertex-edge event would happen.
%     
%     \remlab{ignoring:edge:edge}
% \end{remark}
% 

\subsection{Curve Simplification}
We suggest a straightforward greedy algorithm for curve
simplification, which is sufficient for our purposes.  We comment that
Agarwal \etal \cite{ahmw-nltaa-05} suggested a more aggressive (but
slightly slower and more complicated) simplification algorithm that
can be used instead.

\begin{algorithm}
    Given a polygonal curve $\curveA= \pntA_1 \pntA_2 \pntA_3 \ldots
    \pntA_k$ and a parameter $\sRadius>0$, consider the following
    simplification algorithm: First mark the initial vertex $\pntA_1$
    and set it as the current vertex. Now scan the polygonal curve
    from the current vertex until it reaches the first vertex
    $\pntA_i$ that is in distance at least $\sRadius$ from the current
    vertex. Mark $\pntA_i$ and set it as the current vertex. Repeat
    this until reaching the final vertex of the curve, and also mark
    this final vertex.  Consider the curve that connects only the
    marked vertices, in their order along $\curveA$. We refer to the
    resulting curve $\curveAs = \simpX{\curveA,\sRadius}$ as the
    \emphi{$\sRadius$-simplification} of $\curveA$. Note, that this
    simplification can be computed in linear time.
    
    \alglab{simplification:algorithm}
\end{algorithm}

\begin{remark}
    The simplified curve has the useful property that all its segments
    are of length at least $\sRadius$, except for the last edge that
    might be shorter. For the sake of simplicity of exposition, we
    assume that the last segment in the simplified curve also has
    length at least $\sRadius$.  Our arguments can be easily modified
    to handle this more general case.
\end{remark}

\begin{lemma}
    For any polygonal curve $\curveA$ in $\Re^d$, and $\sRadius \geq
    0$, it holds $\distFr{\MakeSBig\curveA}{\simpX{\curveA,\sRadius}}
    \leq \sRadius$.
    
    \lemlab{simplification:distance}
\end{lemma}

% \InSubmitVer{\vspace{-1cm}} \InSubmitVer{\vspace{1cm}}
\parpic[r]{\includegraphics[scale=0.75]{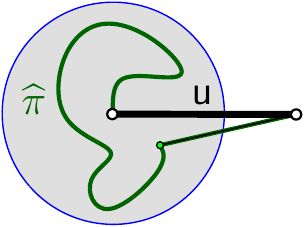}}
\vspace{-15\lineskip}
\begin{proof}
    Consider a segment $\segA$ of $\simpX{\curveA,\sRadius}$ and the
    portion $\subcurveA$ of $\curveA$ that corresponds to it. Clearly,
    all the vertices of $\subcurveA$ are contained inside a ball of
    radius $\sRadius$ centered at the first endpoint of $\segA$
    visited by $\curveA$, except the last vertex of $\subcurveA$. As
    such, one can parameterize $\segA$ and $\subcurveA$, such that
    initially the point stays on the vertex of $\segA$ while visiting
    all vertices of $\subcurveA$ (except the last one), and then
    simultaneously move in sync on $\segA$ and the last segment of
    $\subcurveA$, in such a way that the distance is always at most
    $\sRadius$.
\end{proof}

\section{The Approximation Algorithm}
\seclab{algorithm}

\subsection{Computing the \Relevant Free Space}
\seclab{relevant}

For two curves $\curveA$ and $\curveB$, their \emphi{\relevant free
   space}, denoted by $\FDleq{\delta}(\curveA,\curveB)$, is the set of
all the points of $\FullFDleq{\delta}(\curveA,\curveB)$ that are
reachable from $(0,0)$ by an $(x,y)$-\emph{monotone} path.

The set $\FDleq{\delta}$ has finite descriptive complexity inside each
grid cell, and we need to describe it only for the grid cells that
have non-empty intersection with $\FDleq{\delta}$.  Clearly,
generating only those grid cells is sufficient to decide if there is a
monotone path between $(0,0)$ and $(1,1)$, which is equivalent to
deciding if the \Frechet distance between $\curveA$ and $\curveB$ is
smaller or equal to $\delta$.  In particular, to fully describe
$\FDleq{\delta}$, we will specify the reachability intervals
$\RHCellXY{i}{j}\subseteq \IHCellXY{i}{j}$ and
$\RVCellXY{i}{j}\subseteq\IVCellXY{i}{j}$ for each cell
$\CellXY{i}{j}$, which describe the intersection of $\FDleq{\delta}$
with the top and right boundary of $\CellXY{i}{j}$.  These intervals
contain all the needed information, since $\FDleq{\delta} \cap
\CellXY{i}{j}$ is convex.

The \emphi{complexity} of the \relevant free space, for distance
$\delta$, denoted by $\Nleq{\delta}(\curveA,\curveB)$, is the total
number of grid cells which have non-empty intersection with
$\FDleq{\delta}$.  One can compute this set of cells and extract an
existing monotone path in $O\pth{\Nleq{\delta}(\curveA,\curveB)}$
time, by performing a \BFS of the grid cells that visits only the
\relevant cells.  This yields the following relatively easy result. We
include the details both for the sake of completeness and because the
algorithm we suggest is engagingly simple.

\begin{lemma}
    Given two polygonal curves $\curveA$ and $\curveB$ in $\Re^d$, and
    a parameter $\delta \geq 0$, one can compute a representation of
    $\FDleq{\delta}(\curveA,\curveB)$ in $O\pth{
       \Nleq{\delta}(\curveA,\curveB)}$ time.  Furthermore, one can
    decide if $\distFr{\curveA}{\curveB} \leq \delta$, and if this is
    the case also extract reparametrizations in
    $O(\Nleq{\delta}(\curveA,\curveB))$ time.
    
    \lemlab{graph:fast}
\end{lemma}

\begin{proof}
    We create a directed graph $\Graph$ that has a node $v(i,j)$ for
    every \relevant free space cell $\CellXY{i}{j}$.  With each node
    $v(i,j)$ we store the free space intervals $\IHCellXY{i}{j}$ and
    $\IVCellXY{i}{j}$ as well as the reachability intervals
    $\RHCellXY{i}{j}\subseteq \IHCellXY{i}{j}$ and
    $\RVCellXY{i}{j}\subseteq\IVCellXY{i}{j}$.
    
    Each node $v(i,j)$ can have an outgoing edge to its right and top
    neighbor; an edge between these vertices exists if and only if the
    corresponding reachability interval between them is nonempty.  In
    particular, a monotone path from $(0,0)$ to a point $(x,y) \in
    \CellXY{i}{j}$ in $\FDleq{\delta}$ corresponds to a monotone path
    in the graph $\Graph$ from $v(1,1)$ to $v(i,j)$. Furthermore, any
    such monotone path has exactly $k= i+j-2$ edges on it.
    
    We compute the graph $\Graph$ on the fly by performing a \BFS on
    it, starting from $v(1,1)$, and keeping the invariant that when
    the \BFS visits a node $v(i,j)$ it enqueues the vertices
    $v(i,j+1)$ and $v(i+1,j)$, in this order, to the \BFS queue (if
    they are connected to $v(i,j)$, naturally).
    
    \parpic[l]{\includegraphics[scale=0.5]{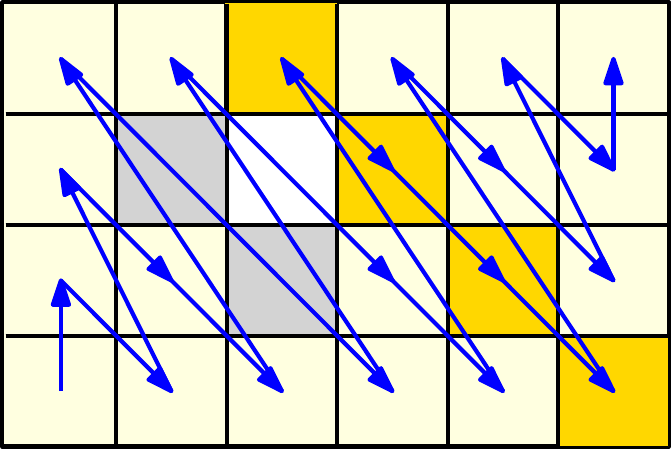}}
    This implies that at any point in time, and for any $k$, the \BFS
    queue contains the nodes on the $k$\th diagonal (i.e., all nodes
    $v(i,j)$ such that $i+j=k-1$) of the diagram sorted from left to
    right. However, the same node might appear twice (consecutively)
    in this queue.
    
    In every iteration, the \BFS dequeues the one or two copies of the
    same node $v(i,j)$ and merges the two copies of the same vertex
    into one if necessary.  Now, the one or two vertices (i.e.,
    $v(i-1,j)$ and $v(i,j-1)$) that have incoming edges to $v(i,j)$
    are known, as are their reachability intervals. Therefore one can
    compute the reachability intervals for $v(i,j)$ in constant time.
    Now, $v(i,j+1)$ is enqueued if and only if the top side of the
    cell $\CellXY{i}{j}$ is reachable by a monotone path (i.e.,
    $\RHCellXY{i}{j}\neq\emptyset$), and $v(i+1,j)$ is enqueued if and
    only if the right side of the cell $\CellXY{i}{j}$ is reachable by
    a monotone path (i.e., $\RVCellXY{i}{j}\neq\emptyset$). Since
    $\FDleq{\delta}(\curveA,\curveB) \cap \CellXY{i}{j}$ is convex and
    of constant complexity, this can be done in constant time.
    
    Clearly, the \BFS takes time linear in the size of $\Graph$ and it
    computes the reachability information for all \relevant free space
    cells of $\FDleq{\delta}(\curveA,\curveB)$.  Now, one can check if
    $(1,1)$ is reachable by inspecting the reachability intervals for
    $\CellXY{\nVertices{\curveA}-1 }{\nVertices{\curveB}-1}$, and
    checking if the top right corner of this cell is monotonically
    reachable from the origin, where $\nVertices{\curveA}$ is the
    number of vertices of the curve $\curveA$. The monotone path
    realizing this can be extracted in linear time, by introducing
    backward edges in the graph and tracing a path back to the origin.
\end{proof}

\begin{observation}
    One can compute all relevant vertex-edge events with radius $\leq
    \delta$ in $O\pth{ \Nleq{\delta}(\curveA,\curveB)}$ time as
    follows.  We compute the graph representation of
    $\FDleq{\delta}(\curveA, \curveB)$ using
    \lemref{graph:fast}. Next, for each \relevant cell consider the
    vertex-edge events at its top and right boundaries and compute
    their event radii. Recall that a cell boundary corresponds to an
    edge from the one curve and a vertex from the other
    curve. Clearly, any cell boundary can be used by the
    reparameterization of width $\leq \delta$, if and only if the
    corresponding event radius is smaller or equal $\delta$.
    
    \obslab{extract:v:e}
\end{observation}

% \subsection{Tools}
\subsection{The Approximate Decision Procedure}
\seclab{decision:procedure}

In the following, we are interested in the maximum complexity of the
\relevant free space when considering any radius $\delta$ and
simplifying the curves with radius $\eps \delta$. The reasons will
become apparent only shortly after, in \lemref{decider} and
\lemref{decider:2}, where we show that the simplification radius
chosen this way enables us to either
\begin{inparaenum}[(i)]
    \item compute a $(1+\eps)$-approximation of the \Frechet distance,
    or
    \item solve the decision problem exactly using the simplified
    curves (see \secref{fixed:simplification}).
\end{inparaenum}

\begin{figure}
    \centerline{\includegraphics{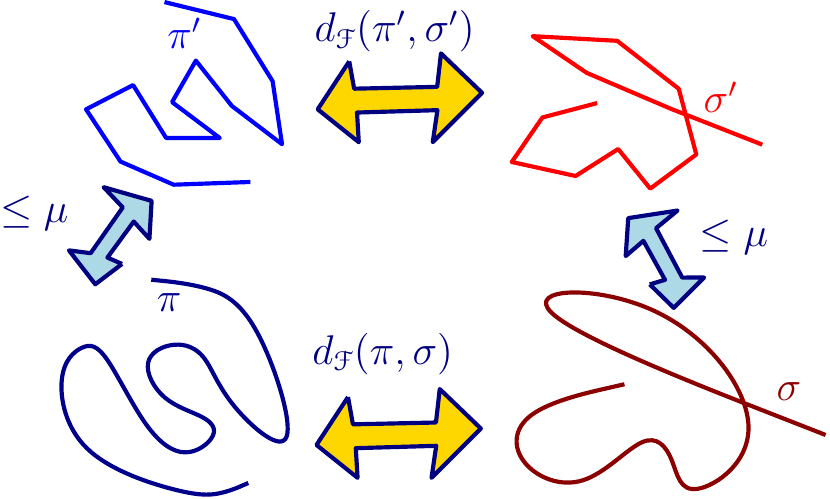}}
    \caption{The idea of the fuzzy decision procedure using
       simplification.  }
    \figlab{fuzzy:decider}
\end{figure}

The idea underlying this approximate decision procedure is depicted in
\figref{fuzzy:decider}. We simplify the two input curves to a
resolution that is (roughly) an $\eps$-fraction of the radius we care
about (i.e., $\delta$), and we then use the exact decision procedure
on these two simplified curves.  Since the \Frechet distance complies
with the triangle inequality and by \lemref{simplification:distance},
we can infer the original distance from this information.  In order
for this approach to work, the complexity of the \relevant free space
for the two simplified curves has to be small.  This notion of
complexity is captured by the following definition.

\begin{defn}
    For two curves $\curveA$ and $\curveB$, let
    \begin{align*}
        \SimpComplexity{\eps}{\curveA}{\curveB} = \max_{\delta \geq 0}
        \;\Nleq{\delta} \pth{\MakeSBig \simpX{\curveA, \eps \delta},
           \simpX{\curveB, \eps \delta }}
    \end{align*}
    be the maximum complexity of the \relevant free space for the
    simplified curves.  We refer to
    $\SimpComplexity{\eps}{\curveA}{\curveB}$ as the
    \emphi{\resemblanceX{\eps}{}} of $\curveA$ and $\curveB$.  In
    order to give a more informative analysis, we will express the
    asymptotic time complexity of our algorithms not in terms of the
    size of the input, but instead use the size of the input and the
    free space complexity of the input as parameters.
    
    We assume that for any $0 < \eps <1$ the following properties hold
    for $\SimpComplexity{\cdot}{\cdot}{\cdot}$.
    
    \medskip
    
    $~$~~~~
    \begin{minipage}{0.8\linewidth}
        \begin{compactenum}[(P1)]
            \item \label{resemblance:a} %
            For any constant $c'\geq 1$, it holds
            $\SimpComplexity{\eps/c'}{\curveA}{\curveB} = O\pth{
               \SimpComplexity{\eps}{\curveA}{\curveB} }$.
            
            \item \label{resemblance:b} %
            $\SimpComplexity{\eps}{\curveA}{\curveB} \leq
            \SimpComplexity{\eps/2}{\curveA}{\curveB} /2$.
        \end{compactenum}
    \end{minipage}
    
    \deflab{simplification:complexity}
\end{defn}

The above properties will hold for all the families of curves we
consider. In \secref{c:packed:curves} we show that
$\SimpComplexity{\eps}{\curveA}{\curveB}$ is a linear function in the
number of vertices of the two curves for a fixed $\eps>0$ if the
curves are sufficiently well-behaved (see for example
\lemref{complexity:l:e:q}). Combining this analysis with the time
complexity analysis of the algorithms will yield near-linear upper
bounds on the running times of these algorithms for the classes of
curves considered.

\begin{remark}
    In the following, when we state the time complexity of our
    algorithms, we always assume that
    $\SimpComplexity{\eps}{\curveA}{\curveB} = \Omega(n)$, where $n$
    is the total number of vertices of $\curveA$ and $\curveB$.

    \remlab{linear:f:s:compl}
\end{remark}

\begin{lemma}
    Let $\curveA$ and $\curveB$ be polygonal curves in $\Re^d$, and
    let $\eps > 0$ and $\delta >0$ be two parameters. Then, the
    algorithm described below output, in $O\pth{
       \SimpComplexity{\eps}{\curveA}{\curveB}}$ time, one of the
    following:
    \begin{compactenum}[\qquad(A)]
        \item ``$\distFr{\curveA}{\curveB} \leq (1+\eps) \delta$'',
        and reparameterizations of $\curveA$ and $\curveB$ of width
        $\leq (1+\eps)\delta$, and this happens if
        $\distFr{\curveA}{\curveB} \leq \delta$.
        
        \item ``$\distFr{\curveA}{\curveB} > \delta$'' if
        $\distFr{\curveA}{\curveB} > (1+\eps)\delta$.
        
        \item If $\distFr{\curveA}{\curveB} \in \left(\delta,
            (1+\eps)\delta\right]$ then the algorithm outputs either
        of the above outcomes.
    \end{compactenum}
    In either case, the statement returned is correct.
    % \footnote{To see
    % why this last statement is not redundant, consider a ``lesser''
    % decision procedure that inside an interval of uncertainty
    % randomly returns one of the two possible answers. Such a
    % decision procedure will sometimes return the wrong answer inside
    % the uncertainty region.}
    
    \lemlab{decider}
\end{lemma}

\begin{proof}
    Set $\sRadius = (\eps/4) \delta$.  Compute in linear time the
    curves $\curveAs= \simpX{\curveA,\sRadius}$ and $\curveBs=
    \simpX{\curveB,\sRadius}$ using \algref{simplification:algorithm}.
    Let $\delta' = \delta + 2\sRadius$ and observe that
    $\nfrac{\sRadius}{\delta'} = {\eps}/ \pth{4 + 2\eps}$. Using
    \lemref{graph:fast} we can decide whether
    $\distFr{\curveAs}{\curveBs} \leq \delta'$ in %
    \begin{align*}
        O\pth{ \Nleq{\delta'}(\curveAs,\curveBs)} %
        =%
        O \pth{ \SimpComplexity{\sRadius /
              \delta'}{\curveA}{\curveB}}%
        =%
        O \pth{ \SimpComplexity{\eps/(4+2\eps)}{\curveA}{\curveB}}%
        =%
        O \pth{ \SimpComplexity{\eps}{\curveA}{\curveB}}%
    \end{align*}
    time, by assumption (P\ref{resemblance:a}).  If so, we output the
    reparameterizations as a proof that
    \begin{align*}
        \distFr{\curveA}{\curveB} \leq \distFr{\curveA}{\curveAs} +
        \distFr{\curveAs}{\curveBs} +\distFr{\curveBs}{\curveB} \leq
        \delta' +2 \sRadius = \delta +4 (\eps/4) \delta =
        (1+\eps)\delta.
    \end{align*}
    
    On the other hand, if $\distFr{\curveAs}{\curveBs} > \delta'$,
    then this implies, by the triangle inequality, that
    \begin{align*}
        \distFr{\curveA}{\curveB} \geq \distFr{\curveAs}{\curveBs} -
        \distFr{\curveA}{\curveAs} - \distFr{\curveBs}{\curveB} >
        \delta' -2 \sRadius = \delta.
    \end{align*}
    Therefore, the algorithm outputs ``$\distFr{\curveA}{\curveB} >
    \delta$'' in this case.
\end{proof}

\subsubsection{How to use the Approximate Decider in a Binary Search}

In order to use \lemref{decider} to perform a binary search for the
\Frechet distance, we can turn the ``fuzzy'' decision procedure into a
precise one as follows.

\begin{lemma}
    Let $\curveA$ and $\curveB$ be two polygonal curves in $\Re^d$,
    and let $1 \geq \eps > 0$ and $\delta >0$ be two parameters. Then,
    there is an algorithm \deciderFr{}$\pth[]{ \curveA, \curveB,
       \delta, \eps }$ that, in $O\pth{
       \SimpComplexity{\eps}{\curveA}{\curveB} }$ time, returns one of
    the following outputs:
    \begin{inparaenum}[(i)]
        \item a $(1+\eps)$-approximation to
        $\distFr{\curveA}{\curveB}$,
        \item $\distFr{\curveA}{\curveB} < \delta$, or
        \item $\distFr{\curveA}{\curveB} > \delta$.
    \end{inparaenum}
    The answer returned is correct.
    
    \lemlab{decider:2}
\end{lemma}

\begin{proof}
    Let $\delta' = \delta/(1+\eps')$, for $\eps' = \cSimpB \eps$,
    $\cSimpB=\cSimpBVal$. We run the algorithm of \lemref{decider}
    with parameters $\delta$ and $\eps'$.  If the call returns
    ``$\distFr{\curveA}{\curveB} > \delta$'', then we return this
    result.
    
    Otherwise, we call \lemref{decider} with parameters $\delta'$ and
    $\eps'$. If it returns that ``$\distFr{\curveA}{\curveB} \leq
    (1+\eps')\delta'$'' then $\distFr{\curveA}{\curveB} \leq
    (1+\eps')\delta' = \delta$, and we return this result.
    
    The only remaining possibility is that the two calls returned
    ``$\distFr{\curveA}{\curveB} \leq (1+\eps')\delta$'' and
    ``$\distFr{\curveA}{\curveB} > \delta'$''.  But then we have found
    the required approximation.  Therefore, the resulting
    approximation factor of the reparameterizations returned by the
    call with $\delta$ is $\ds \leq \frac{(1+\eps') \delta}{ \delta'}
    = (1+ \cSimpB \eps)^2 < (1+\eps)$ as can be easily verified, since
    $0 < \eps \leq 1$.
\end{proof}

\subsection{Searching for the \Frechet Distance}
\seclab{search}

\subsubsection{Searching in a Fixed Interval}
\seclab{search:fixed:interval}

It is now straightforward to perform a binary search on an interval
$[\alpha,\beta]$ to approximate the value of the \Frechet distance, if
it falls inside this interval. Indeed, partition this interval into
subintervals of length $\eps \alpha$ and perform a binary search to
find the interval that contains the \Frechet distance. There are
$O(\beta/\eps\alpha)$ intervals, and this would require $O( \log
(\beta / \eps \alpha) )$ calls to \deciderFr{}. By using exponential
subintervals, one can do slightly better, as testified by the
following lemma.

\begin{lemma}
    % Given two curves $\curveA$ and $\curveB$ in $\Re^d$ of total
    % complexity $n$, a parameter $1 \geq \eps>0$, and an interval
    % $[\alpha,\beta]$, one can compute a $(1+\eps)$-approximation to
    % $\distFr{\curveA}{\curveB}$ if $\distFr{\curveA}{\curveB} \in
    % [\alpha,\beta]$, or report that $\distFr{\curveA}{\curveB}
    % \notin [\alpha,\beta]$.  The algorithm, denoted by
    % \intervalFr{}$\pth[]{ \curveA, \curveB, [\alpha, \beta], \eps
    % }$, takes $\ds O\pth{ \SimpComplexity{\eps}{\curveA}{\curveB}
    %    \log \frac{\log (\beta/\alpha)}{\eps} }$ time.
    Given two curves $\curveA$ and $\curveB$ in $\Re^d$, a parameter
    $1 \geq \eps>0$, and an interval $[\alpha,\beta]$, one can perform
    a binary search in $[\alpha,\beta]$ and obtain a
    $(1+\eps)$-approximation to $\distFr{\curveA}{\curveB}$ if
    $\distFr{\curveA}{\curveB} \in [\alpha,\beta]$, or report that
    $\distFr{\curveA}{\curveB} \notin [\alpha,\beta]$.  The algorithm,
    denoted by \intervalFr{}$\pth[]{ \curveA, \curveB, [\alpha,
       \beta], \eps }$, takes $\ds O\pth{ \log \frac{\log
          (\beta/\alpha)}{\eps} }$ calls to \deciderFr{}.

    \lemlab{Frechet:naive}
\end{lemma}

\begin{proof}
    Let $\alpha_i = \alpha(1+\eps)^i$ for $i=0, \ldots, M =
    \floor{\log_{1+\eps} (\beta/\alpha)}$ and
    $\alpha_{M+1}=\beta$. Perform a binary search, using
    \deciderFr{}$(\curveA, \curveB, \delta,\eps)$ to find the two
    values $\alpha_i$ and $\alpha_{i+1}$ such that
    $\alpha_i\leq\delta=\distFr{\curveA}{\curveB}\leq\alpha_{i+1}$. Since
    $\alpha_{i+1} = (1+\eps)\alpha_i$, we conclude that we found the
    required approximation.
    
    It might be that during this procedure one of the calls to
    \deciderFr{}$(\curveA, \curveB, \delta,\eps)$ found the required
    approximation, and in this case we abort the binary search and
    just return this approximation.
    
    This process requires $O( \log M) = O\pth{\log
       \log_{1+\eps}(\beta/\alpha)}$ calls to \deciderFr. Observe that
    \begin{align*}
        M = \log_{1+\eps} \frac{\beta}{\alpha} = \frac{\ln
           (\beta/\alpha)}{ \ln( 1+\eps ) } =O \pth{ \frac{1}{\eps}
           \log \frac{\beta}{\alpha}}.
    \end{align*}
    Indeed, $e^{x/2} \leq 1+x \leq e^x$ for $x \in [0,1]$, and this
    implies that $x/2 \leq \ln (1+x) \leq x$, which is the inequality
    used above.
\end{proof}

\subsubsection{Searching over Events}

Clearly, the procedure \intervalFr{}$\pth[]{ \curveA, \curveB,[\alpha,
   \beta],\eps}$ alone does not suffice to solve our main problem,
since the interval of distances we are searching over might have
arbitrarily large ``spread'' (i.e., $\log \beta/\alpha$ might be
arbitrarily large).  However, the \Frechet distance must be
sufficiently close to a free space event in one of the ``approximate''
diagrams, i.e., a free space diagram of the two simplified
curves. Thus, we can identify two kinds of critical values to search
over, which are candidate values for the approximate \Frechet
distance.  These are the events where
\begin{inparaenum}[(i)]
    \item the simplification of an input curve changes, or
    \item the reachability within the approximate free space diagram
    changes (i.e., a free space event; see \secref{f:s:events}).
\end{inparaenum}

The traditional solution to overcome this problem is to use parametric
search. However, in our case, since we are only interested in
approximation, we can use a simpler, ``approximate'', search.  It is
sufficient to search over a set of values which approximate the event
values by a constant factor, since we will use \lemref{Frechet:naive}
to refine the resulting search interval in the main algorithm.  Note,
for instance, that we can easily use this lemma to turn a constant
factor approximation of the \Frechet distance into a
$(1+\eps)$-approximation.

\begin{algorithm}
    Let \approxBinarySearch{}$(\curveA,\curveB,Z)$ denote the
    algorithm that performs a binary search over the values of $Z$, to
    compute the atomic interval of $Z$ that contains the \Frechet
    distance between $\curveA$ and $\curveB$. This procedure uses
    $\deciderFr$ (\lemref{decider:2}) to perform the decisions during
    the search.
    
    \alglab{a:binary:search}
\end{algorithm}

% \subsubsection{Simplification Events}

\subsubsection{Searching over Simplifications }
\seclab{search:simplifications}

Consider the events when the simplified curves change, see
\algref{simplification:algorithm}. Consider the set of all pairwise
distances between vertices of $\curveA$ and $\curveB$. Observe that it
breaks the real line into $\binom{n}{2} + 1$ atomic intervals, such
that in each such interval the simplification does not change. Thus
$\simpX{\curveA,\sRadius}$ (resp.  $\simpX{\curveB,\sRadius}$) might
result in $O(n^2)$ different curves depending on the value of
$\sRadius$, where $n$ is the total number of vertices of $\curveA$ and
$\curveB$.  As a first step we would therefore like to use
\algref{a:binary:search} to perform a binary search over those
distances to find the atomic interval that contains the required
\Frechet distance.  Naively, this would require us to perform distance
selection. However, it is believed that exact distance selection
requires $\Omega\pth{n^{4/3}}$ time in the worst case
\cite{e-rcsgp-95}. To overcome this we will perform an approximate
distance selection, as suggested by Aronov \etal \cite{ahkww-fdcr-06}.

\begin{lemma}
    Given a set $\PntSet$ of $n$ points in $\Re^d$.  Then, one can
    compute in $O(n \log n)$ time a set $Z$ of $O(n)$ numbers, such
    that for any $y \in \Pairwise{\PntSet}$, there exist numbers $x,x'
    \in Z$ such that $x \leq y \leq x' \leq 2x$.  Let
    $\approxDistances(\PntSet)$ denote this algorithm.
    
    \lemlab{all:distances}
\end{lemma}

\begin{proof}
    Compute an $8$-well-separated pairs decomposition of
    $\PntSet$. Using the algorithm of Callahan and Kosaraju
    \cite{ck-dmpsa-95} this can be done in $O(n \log n)$ time, and
    results in a set of pairs of subsets $\brc{ (X_1, Y_1), \ldots,
       (X_m, Y_m)}$, where $m=O(n)$, such that for any two points
    $\pntA, \pntB \in \PntSet$ there exists a pair $(X_i, Y_i)$ in the
    above decomposition, such that:
    \begin{inparaenum}[(i)]
        \item $\pntA \in X_i$ and $\pntB \in Y_i$ (or vice versa), and
        \item $\max(\diameterX{X_i}, \diameterX{Y_i}) \leq
        \min_{\pntA_i \in X_i,\pntB_i \in Y_i}
        \distX{\pntA_i}{\pntB_i}/8.$
	% \distX{\pntA}{\pntB}/8$.
    \end{inparaenum}
    
    This implies that the distance of any pair of points in $X_i$ and
    $Y_i$, respectively, are the same up to a small constant.  As
    such, for every pair $(X_i, Y_i)$, for $i=1,\ldots, m$, we pick
    representative points $\pntA_i \in X_i$ and $\pntB_i \in Y_i$, and
    set $\ell_i = (3/4)\distX{\pntA_i}{\pntB_i}$. Let $Z = \brc{
       \ell_1, \ldots, \ell_m, 2\ell_1, \ldots, 2\ell_m}$ be the
    computed set of values.
    
    Consider any pair of points $\pntA, \pntB \in \PntSet$.  For the
    specific pair $(X_i, Y_i)$ that contains the pair of points
    $\pntA$ and $\pntB$ that we are interested in, we have that
    $\ell_i = (3/4)\distX{\pntA_i}{\pntB_i} \leq
    \distX{\pntA_i}{\pntB_i} - \diameterX{X_i} - \diameterX{Y_i} \leq
    \distX{\pntA}{\pntB} \leq \distX{\pntA_i}{\pntB_i} +
    \diameterX{X_i} + \diameterX{Y_i} \leq
    (5/4)\distX{\pntA_i}{\pntB_i}\leq 2 \ell_i$, thus establishing the
    claim.
\end{proof}

\subsubsection{Monotonicity Events}

The following lemma testifies that the radius of a monotonicity event
must be ``close'' to either a vertex-edge event or to the distance
between two vertices. Since we will approximate the vertex-vertex
distances and perform a binary search over them, this implies that we
further only need to consider vertex-edge events.  Furthermore, by
\obsref{extract:v:e}, the number of those vertex-edge events which
remain in the resulting search range can be bounded by the complexity
of the reachable free space.

\begin{lemma}
    Let $x$ be the radius of a monotonicity event involving vertices
    $\pntA, \pntB$ and a segment $\segA$. Then there exists a number
    $y$ such that $y/2 \leq x \leq 3y$, and $y$ is either in $\PW =
    \Pairwise{\VertexSet{\curveA} \cup \VertexSet{\curveB}}$ or $y$ is
    the radius of a vertex-edge event.
    
    \lemlab{monotone:not:boring}
\end{lemma}

\begin{proof}
    Let $\pntC$ be the intersection point of $\SphereX{\pntA, x} \cap
    \SphereX{\pntB, x}$ which lies on $\segA$.  Let $\pntA'$ (resp.
    $\pntB'$) be the closest point on $\segA$ to $\pntA$ (resp.
    $\pntB$).
    
    \ParaWPic{\includegraphics[scale=1]{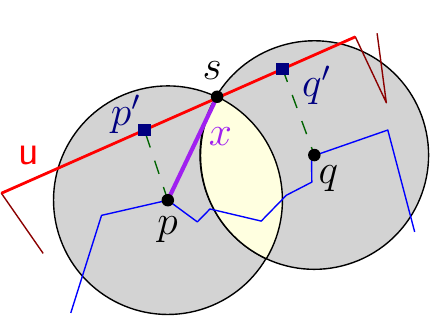}}
    { % Otherwise,
       
       Clearly $\distX{\pntA'}{\pntB'} \leq \distX{\pntA}{\pntB}$
       (since the projection onto the nearest neighbor of a convex set
       is a contraction), and since $\pntA'\in\BallX{\pntA, x}$ and
       $\pntB'\in\BallX{\pntB,x}$,
       % 
       % and since the $\eps$-environment around $\pntC$ on $\segA$ is
       % covered by $\BallX{\pntA, x} \cup \BallX{\pntB, x}$,
       %
       the point $\pntC$ lies on the segment $\pntA' \pntB'$.
       
       This implies that $x = \distX{\pntA}{\pntC} \leq
       \distX{\pntA}{\pntA'} + \distX{\pntA'}{\pntC} \leq
       \distX{\pntA}{\pntA'} + \distX{\pntA'}{\pntB'} \leq
       \distX{\pntA}{\pntA'} + \distX{\pntA}{\pntB}$, by the triangle
       inequality. %
    }%
    
    A similar argument implies that
    \begin{align*}
        x = \distX{\pntA}{\pntC} \geq \distX{\pntA}{\pntA'} -
        \distX{\pntA'}{\pntC} \geq \distX{\pntA}{\pntA'} -
        \distX{\pntA'}{\pntB'} \geq \distX{\pntA}{\pntA'} -
        \distX{\pntA}{\pntB}.
    \end{align*}
    
    If $\distX{\pntA}{\pntA'} \geq 2\distX{\pntA}{\pntB}$ then the
    above implies that $x \in [1/2,3/2]\distX{\pntA}{\pntA'}$.  If
    $\pntA'$ is an endpoint of $\segA$ then $\distX{\pntA}{\pntA'}$ is
    in $\PW$. Otherwise, $\distX{\pntA}{\pntA'}$ is the radius of the
    vertex-edge event between $\pntA$ and $\segA$.  In either case,
    this implies the claim.
    
    If $\distX{\pntA}{\pntA'} \leq 2\distX{\pntA}{\pntB}$ then $ x =
    \distX{\pntA}{\pntC} %
    \leq %
    \distX{\pntA}{\pntA'} + \distX{\pntA}{\pntB} \leq
    2\distX{\pntA}{\pntB} + \distX{\pntA}{\pntB} = 3
    \distX{\pntA}{\pntB}$, and of course $\distX{\pntA}{\pntB} \in
    \PW$.  Now, the two balls of radius $x$ centered at $\pntA$ and
    $\pntB$, respectively, cover the segment $\pntA \pntB$, and we
    have that $\distX{\pntA}{\pntB}/2 \leq x$, which implies the
    claim.
\end{proof}

\subsubsection{Searching with a Fixed Simplification}
\seclab{fixed:simplification}

Assume that we have found simplifications $\curveC$ and $\curveD$,
such that the \Frechet distance of those curves yields the desired
$(1+\eps)$-approximation. Clearly, an approximation of
$\distFr{\curveC}{\curveD}$ suffices for our result.  To this end, let
$\intervalExactFr{}( \curveA, \curveB, [\alpha, \beta], \eps)$ be the
variant of \intervalFr{} from \lemref{Frechet:naive} that uses
\lemref{graph:fast} directly instead of calling \deciderFr.  This
version searches for the \Frechet distance in the given interval, but
does not perform simplification before calling the decision
procedure. It returns a $(1+\eps)$-approximation of the \Frechet
distance, given that it is contained in this interval. Note that
correctness and running time of \lemref{Frechet:naive} are not
affected by this modification.

\begin{lemma}
    Let $\curveC$ and $\curveD$ be two given curves in $\Re^d$, with
    total complexity $n$, and let $[\loLimit, \hiLimit]$ be an
    interval, such that
    \begin{inparaenum}[(i)]
        \item $\distFr{\curveC}{\curveD} \in [\loLimit, \hiLimit]$,
        and
        \item there is no value of $\PW =
        \Pairwise{\VertexSet{\curveC} \cup \VertexSet{\curveD}}$ in
        the interval $[\loLimit, \hiLimit]$.
    \end{inparaenum}
    Then, for $\eps > 0$, one can $(1+\eps)$-approximate
    $\distFr{\curveC}{\curveD}$ and compute reparametrizations in
    $O\pth{ (n + N) \log (\nfrac{N}{\eps}) }$ time, where $N =
    \Nleq{\hiLimit}(\curveC,\curveD)$.
    
    Let $\approxFrDirect( \curveC, \curveD, [\loLimit, \hiLimit],
    \eps)$ denote this algorithm.
    
    \lemlab{direct}
\end{lemma}

\begin{proof}
    For two real numbers $x,y >0$, we define $\ratioX{x}{y} =
    \max(x,y)/ \min(x,y)$.
    
    Compute $\FDleq{\hiLimit}(\curveC, \curveD)$, using
    \lemref{graph:fast}.  Next, using \obsref{extract:v:e}, compute
    from $\FDleq{\hiLimit}(\curveC, \curveD)$ the set $\veEvents$ of
    all the radii of the vertex-edge events of $\curveC$ and $\curveD$
    with radius at most $\hiLimit$.  Next, we sort $\veEvents$, and
    perform a binary search over $\veEvents$, using
    \lemref{graph:fast}, for the atomic interval
    $\Interval=[\alpha,\beta]$ of $\veEvents$ that contains the
    \Frechet distance $\distFr{\curveC}{\curveD}$.  Next, call
    \intervalExactFr{}$( \curveC, \curveD, [\alpha, 4\alpha], \eps)$
    and \intervalExactFr{}$( \curveC, \curveD, [\beta/4, \beta],
    \eps)$.  We claim that one of these two searches performed on the
    respective intervals will discover two consecutive values $x$ and
    $(1+\eps)x$, such that the two corresponding calls to the
    algorithm of \lemref{Frechet:naive} imply that
    $\distFr{\curveC}{\curveD} \in [x, (1+\eps)x]$.
    
    Indeed, the interior of $[\alpha,\beta]$ does not contain any
    value in $\PW$ or a radius of a vertex-edge event of $\curveC$ and
    $\curveD$. Therefore, the interval $[\alpha,\beta]$ might contain
    only monotonicity events of $\curveC$ and $\curveD$.  By
    \lemref{monotone:not:boring}, for a monotonicity event with radius
    $\Radius$ there exists a $y \in \veEvents \cup \PW$, such that
    $\ratioX{\Radius}{y} \leq 3$.  But since there is no value of
    $\veEvents \cup \PW$ in the interior of $[\alpha,\beta]$, and
    therefore, for any $\Radius'' \in [4\alpha, \beta/4]$ and $y'' \in
    \veEvents \cup \PW$, we have that $\ratioX{\Radius''}{y''} \geq
    4$.
    
    We conclude that no monotonicity event, vertex-edge event, or
    value of $\PW$ lies in the interval $[4\alpha, \beta/4]$. Since
    the \Frechet distance must be equal to one such value, it follows
    that $\distFr{\curveC}{\curveD} \notin (4\alpha, \beta/4)$, but
    this implies that either $\distFr{\curveC}{\curveD} \in [\alpha,
    4\alpha]$ or $\distFr{\curveC}{\curveD} \in [\beta/4, \beta]$. In
    either case, the above algorithm would have found the approximate
    distance.
    
    Computing and sorting the set of vertex-edge events takes $O(N
    \log N)$ time by \obsref{extract:v:e}.  The binary search requires
    $O( \log \cardin{\veEvents} )$ calls to the algorithm of
    \lemref{graph:fast}. The two calls to \intervalExactFr{} require
    $O(\log(1/\eps))$ calls to \lemref{graph:fast}.  Now, observe that
    all these calls to the algorithm of \lemref{graph:fast} are done
    with values of $\delta \leq \hiLimit$. Thus the complexity of the
    \relevant free space is bounded (up to a constant factor) by the
    number of vertex-edge events of values $\leq \hiLimit$, and this
    number is bounded by $\cardin{\veEvents}$. Therefore, a call to
    \lemref{graph:fast} takes $O\pth{\cardin{\veEvents}}$ time.  Thus,
    the overall running time is $O\pth{ \pth{ n +\cardin{\veEvents}}
       \log (\cardin{\veEvents}/\eps) } $, and by definition
    $\cardin{\veEvents}= O\pth{ \Nleq{\hiLimit}(\curveC,\curveD)}$.
\end{proof}

\subsection{The Approximation Algorithm}
\seclab{the:algorithm}

The resulting approximation algorithm is depicted in
\figref{algorithm}.  It will be used by the final approximation
algorithm as a subroutine. We first analyze this basic algorithm. We
will then show how to use it, in \lemref{improved:running:time} below,
to get a faster approximation algorithm.  The algorithm depicted in
\figref{algorithm} performs numerous calls to \deciderFr, with
approximation parameter $\eps>0$.  If any of these calls discover the
approximate distance, then the algorithm immediately stops and returns
the approximation. Therefore, at any point in the execution of the
algorithm, the assumption is that all previous calls to \deciderFr
returned a direction where the optimal distance must lie.  In
particular, a call to \intervalFr{}$\pth[]{ \curveA, \curveB,
   \Interval, \eps}$, would either find the approximate distance in
the interval $\Interval$ and return immediately, or the desired value
is outside this interval.

\begin{figure}[t]
    \begin{center}
        \fbox{~~~~~~\begin{minipage}{0.90\linewidth}
               \hspace{-0.7cm}\approxFr{}($\curveA$, $\curveB$, $\eps$
               )
               \begin{compactenum}[(A)]
                   \item $\PntSet = \VertexSet{\curveA} \cup
                   \VertexSet{\curveB}$
                   
                   \item $\sEvents \leftarrow
                   \approxDistances(\PntSet)$
                   (\lemref{all:distances}).
                   
                   \item
                   \label{simplifications:search:n}
                   $[\alpha,\beta] \leftarrow \approxBinarySearch(
                   \curveA, \curveB, \sEvents, \eps)$
                   (\algref{a:binary:search}).
                   
                   \item
                   \label{search:fringe1:n}
                   Call \intervalFr{}$\pth[]{ \curveA, \curveB,
                      [\alpha, 4\alpha'], \eps }$, where $\alpha'
                   =(30/\eps)\alpha$ (\lemref{Frechet:naive}).
                   
                   \item
                   \label{search:fringe2:n}
                   Call \intervalFr{}$\pth[]{ \curveA, \curveB,
                      [\beta'/4, \beta], \eps }$, where $\beta' =
                   \beta/3$.
                   
                   \item
                   \label{fixed:curves:n}
                   Let $\curveAs = \simpX{\curveA, \sRadius}$ and
                   $\curveBs = \simpX{\curveB, \sRadius}$, for
                   $\sRadius = 3\alpha$
                   (\algref{simplification:algorithm})
                   \item
		   \label{search:direct}
		   $\approxDist \leftarrow \approxFrDirect( \curveAs,
                   \curveBs, [\alpha',\beta'], \eps/4)$
                   (\lemref{direct}).

                   \item
		   \label{final:step}
		   Compute and return the resulting
                   reparameterizations of $\curveA$ and $\curveB$ and
                   their width as the approximation.
                   % Return the reparameterization of $\curveA$ and
                   % $\curveB$ resulting from chaining the
                   % reparameterization of $\curveA \Leftrightarrow
                   % \curveAs \Leftrightarrow \curveBs \Leftrightarrow
                   % \curveB$ and its width as the approximation.
               \end{compactenum}
           \end{minipage}~~}
    \end{center}
    \vspace{-0.6cm}
    \caption{The basic approximation algorithm.}
    \figlab{algorithm}
\end{figure}

\subsubsection{Correctness}

\begin{lemma}
    Given two polygonal curves $\curveA$ and $\curveB$, and a
    parameter $1 > \eps > 0 $, the algorithm
    $\approxFr(\curveA,\curveB,\eps)$ computes a
    $(1+\eps)$-approximation to $\distFr{\curveA}{\curveB}$.
\end{lemma}

\begin{proof}
    If the algorithm found the approximation before step
    (\ref{fixed:curves:n}), then clearly it is the desired
    approximation, and we are done. (In particular, this must be the
    case if $4\alpha' > \beta'/4$.)
    
    Otherwise, because of (\ref{simplifications:search:n}), we know
    that $\distFr{\curveA}{\curveB} \in [\alpha,\beta]$.  By steps
    (\ref{search:fringe1:n}) and (\ref{search:fringe2:n}) it must be
    that $\distFr{\curveA}{\curveB} \in [4\alpha', \beta'/4]$.  Since
    $\sRadius = 3\alpha = (\eps/10) \alpha' \leq \beta'/4$, it
    follows, by the triangle inequality, that
    \begin{align*}
        \distFr{\curveAs}{\curveBs} \leq \distFr{\curveAs}{\curveA} +
        \distFr{\curveA}{\curveB} + \distFr{\curveB}{\curveBs} \leq
        2\sRadius + \beta'/4 < \beta'.
    \end{align*}
    A similar argument shows that $\distFr{\curveAs}{\curveBs} >
    \alpha'$.  Hence, the algorithm of \lemref{direct} can be applied
    to $\curveAs$ and $\curveBs$ for the range $[\alpha',\beta']$, as
    $\distFr{\curveAs}{\curveBs} \in [\alpha',\beta']$.
    
    Now, by \lemref{direct}, we have that the value $\delta$ resulting
    from step (\ref{search:direct}), is contained in the interval
    $[\distFr{\curveAs}{\curveBs}, (1+\eps/4)
    \distFr{\curveAs}{\curveBs}]$.  By the triangle inequality we
    conclude that the returned \Frechet distance is
    \begin{align*}
        \Delta%
        &\leq %
        \distFr{\curveA}{\curveAs} + \approxDist +
        \distFr{\curveB}{\curveBs}%
        \leq%
        \distFr{\curveA}{\curveAs} +
        (1+\eps/4)\distFr{\curveAs}{\curveBs} +
        \distFr{\curveBs}{\curveB}%
        \\
        &\leq%
        (1+\eps/4) \pth{ 2\sRadius + \distFr{\curveA}{\curveB} +
           2\sRadius}
        \leq %
        5 \sRadius + (1+\eps/4) \distFr{\curveA}{\curveB} \leq
        (1+\eps)\distFr{\curveA}{\curveB},
    \end{align*}
    since $5\sRadius = 15\alpha = (\eps/2) (30/\eps) \alpha = (\eps/2)
    \alpha' \leq (\eps/2)\distFr{\curveA}{\curveB}$.
    
    Note that $\Delta \geq \distFr{\curveA}{\curveB}$ since it is the
    width of a specific reparameterization between the two curves.
\end{proof}

\subsubsection{Running Time}

\begin{lemma}
    For any $x,y \in (2\alpha, \beta/2)$, we have $\simpX{\curveA,
       x}=\simpX{\curveA, y} $ and $\simpX{\curveB, x} =
    \simpX{\curveB, y}$.
    
    \lemlab{simplification:stable}
\end{lemma}
\begin{proof}
    Indeed, the interval $(\alpha,\beta)$ does not contain any value
    of $\sEvents$. As such, by \lemref{all:distances}, $(2\alpha,
    \beta/2)$ does not contain any value of the pairwise distances
    between vertices of the vertex set of $\curveA$ and $\curveB$
    which implies that the simplification is the same for any value
    inside this interval.
\end{proof}

\begin{lemma}
    Given two polygonal curves $\curveA$ and $\curveB$ with a total of
    $n$ vertices in $\Re^d$, and a parameter $1 > \eps > 0$, the
    running time of \approxFr{}$\pth[]{\curveA, \curveB, \eps}$ is
    $\ds O\pth{ \SimpComplexity{\eps}{\curveA}{\curveB} \log n}$.
    
    \lemlab{main:1:n}
\end{lemma}

\begin{proof}
    Computing $\sEvents$ (and sorting it) takes $O(n \log n)$ time by
    \lemref{all:distances}. Steps (\ref{simplifications:search:n}),
    (\ref{search:fringe1:n}) and (\ref{search:fringe2:n}) perform $O(
    \log n + \log (1/\eps) ) = O( \log n)$ calls to \deciderFr, by
    \lemref{Frechet:naive}.  (Here, we assume that $\eps = \Omega(
    1/n)$. If $\eps < 1/n$ then we can just use the algorithm of Alt
    and Godau \cite{ag-cfdbt-95} since its running time is faster than
    our approximation algorithm in this case.)  Each call to
    \deciderFr{} takes $O\pth{ \SimpComplexity{\eps}{\curveA}{\curveB}
    } $ time, so overall this takes $O(
    \SimpComplexity{\eps}{\curveA}{\curveB} \log n )$ time.  Computing
    the simplifications in step (\ref{fixed:curves:n}) with
    \algref{simplification:algorithm} takes $O(n)$ time.
    
    By \lemref{direct}, a call to $\approxFrDirect( \curveAs,
    \curveBs, [\alpha',\beta'], \eps/4)$ takes $T = O( (n+ N) \log
    (N/\eps))$ time, with $N=\Nleq{\beta'}(\curveAs, \curveBs)$.  Now,
    $3\alpha$ and $\beta'$ are both inside the interval $(2\alpha,
    \beta/2)$, and as such, by \lemref{simplification:stable}, we have
    that $\curveAs = \simpX{\curveA, 3\alpha} = \simpX{\curveA,
       \beta'}$ and $\curveBs = \simpX{\curveB,3\alpha} =
    \simpX{\curveB,\beta'}$.  Therefore, we have that
    \begin{align*}
        N = \Nleq{\beta'}(\curveAs, \curveBs) =
        \Nleq{\beta'}\pth{\simpX{\curveA,\beta'},
           \simpX{\curveB,\beta'}}
        % \leq \max_{\delta \geq 0}
        % \Nleq{\delta}\pth{\simpX{\curveA,\delta},
        % \simpX{\curveB,\delta}},
        \leq \SimpComplexity{1}{\curveA}{\curveB}.
    \end{align*}
    Thus, step (\ref{search:direct}) takes $T = O(
    \SimpComplexity{1}{\curveA}{\curveB} \log (
    \SimpComplexity{1}{\curveA}{\curveB} n/\eps) ) =O(
    \SimpComplexity{1}{\curveA}{\curveB} \log n )$, time since
    $\SimpComplexity{1}{\curveA}{\curveB} \leq n^2$ and $\eps =
    \Omega( 1/n )$.  Observe that $
    \SimpComplexity{1}{\curveA}{\curveB} \leq
    \SimpComplexity{\eps}{\curveA}{\curveB}$ for $\eps \leq 1$.
    
    Finally, in order to compute the resulting reparameterizations in
    step (\ref{final:step}), we compute the reparametrizations of
    $\curveA$ and $\curveAs$ (resp. $\curveB$ and $\curveB'$) as
    described in the proof of \lemref{simplification:distance} and
    chain them with the reparameterizations of the simplified curves,
    which we obtained from step (\ref{search:direct}).  Clearly, this
    and computing the resulting width takes $O(n)$ time.  Note that by
    the assumption in \remref{linear:f:s:compl} the term
    $\SimpComplexity{\eps}{\curveA}{\curveB}$ dominates over $O(n)$.    
\end{proof}

The running time of \lemref{main:1:n} can be slightly improved.

\begin{lemma}
    The algorithm \approxFr{} depicted in \figref{algorithm}
    % \lemref{main:1:n}
    can be modified to run in time $\ds O(
    \SimpComplexity{\eps}{\curveA}{\curveB} +
    \SimpComplexity{1}{\curveA}{\curveB} \log n)$ (see
    \defref{simplification:complexity}).

    \lemlab{improved:running:time}
\end{lemma}
\begin{proof}
    Use \lemref{main:1:n}, with $\eps_0 = 1/2$, to get a
    $2$-approximation $\zeta$ for the \Frechet distance between
    $\curveA$ and $\curveB$. This takes $O\pth{
       \SimpComplexity{1}{\curveA}{\curveB} \log n }$ time. Let
    $\Interval_0 = [\zeta,2\zeta]$ be the corresponding interval that
    contains the distance. We could call \intervalFr{}$\pth[]{\curveA,
       \curveB, \Interval_0, \eps }$ and get a
    $(1+\eps)$-approximation in $O\pth{
       \SimpComplexity{\eps}{\curveA}{\curveB} \log \frac{1}{\eps} +
       \SimpComplexity{1}{\curveA}{\curveB} \log n }$ time.
    
    One can do better by starting with a ``large'' $\eps$ and
    decreasing it during the binary search for the right value
    performed by \intervalFr. This is a standard idea and it was also
    used by Aronov and Har-Peled \cite{ah-adrp-08}.
    
    Indeed, assume that in the beginning of the $i$\th step, we know
    that the required \Frechet distance lies in an interval
    $\Interval_{i-1} = [\alpha_{i-1},\beta_{i-1}]$ and $\beta_{i-1} -
    \alpha_{i-1}= \lenX{\Interval_0} \eps_{i-1}$, where $\eps_{i-1} =
    1/2^{i-1}$.
    
    Let $\Delta_{i-1} = \lenX{\Interval_{i-1}} = \beta_{i-1}-
    \alpha_{i-1}$, and let $x_{i,j} = \alpha_{i-1} + j\Delta_{i-1}/4$,
    for $j=0,1,2,3,4$.  Call the procedure $\deciderFr$ on three
    values $x_{i,1}$, $x_{i,2}$, and $x_{i,3}$, with the approximation
    parameter being $\constA\eps_i$, for $\constA>0$ being a
    sufficiently small constant. Based on the outcome of these three
    calls, we can determine in constant time which of the three
    intervals $\wInterval_{i,1} = [x_{i,0}, x_{i,2}]$,
    $\wInterval_{i,2} = [x_{i,1}, x_{i,3}]$, or $\wInterval_{i,3} =
    [x_{i,2}, x_{i,4}]$ must contain the \Frechet distance. We set
    this interval to be $\Interval_i$.
    
    We repeat this process for $M$ steps, where
    $M=\ceil{\lg{1/\eps}}$. It is easy to verify that the final
    interval now provides the required approximation. The running time
    of this algorithm is %
    $%
    O\pth{ \SimpComplexity{1}{\curveA}{\curveB} \log n +
       \sum_{i=1}^{M} \SimpComplexity{\eps_i}{\curveA}{\curveB} } %
    $.  Now, by assumption (P\ref{resemblance:b}) (see
    \defref{simplification:complexity}), we have
    \begin{align*}
        O\pth{ \sum_{i=1}^{M}
           \SimpComplexity{\eps_i}{\curveA}{\curveB} } %
        &=%
        O\pth{ \sum_{i=1}^{M} \frac{1}{2^{M-i}}
           \SimpComplexity{\frac{\eps_i}{2^{M-i}}}{\curveA}{\curveB}
        } %
        =%
        O\pth{ \SimpComplexity{\eps}{\curveA}{\curveB} {\sum_{i=1}^{M}
              \frac{1}{2^{M-i}}} } %
        \\&=%
        O\pth{ \SimpComplexity{\eps}{\curveA}{\curveB} },
    \end{align*}
    and this implies the claim.
\end{proof}

\paragraph{The Result.}
Putting the above together, we get the following result.

\begin{theorem}
    Given two polygonal curves $\curveA$ and $\curveB$ with a total of
    $n$ vertices in $\Re^d$, and a parameter $1 > \eps > 0$, one can
    $(1+\eps)$-approximate the \Frechet distance between $\curveA$ and
    $\curveB$ in $\ds O( \SimpComplexity{\eps}{\curveA}{\curveB} +
    \SimpComplexity{1}{\curveA}{\curveB} \log n)$ time (see
    \defref{simplification:complexity}).
    
    \thmlab{main}
\end{theorem}

\parpic[r]{\includegraphics[scale=0.8]{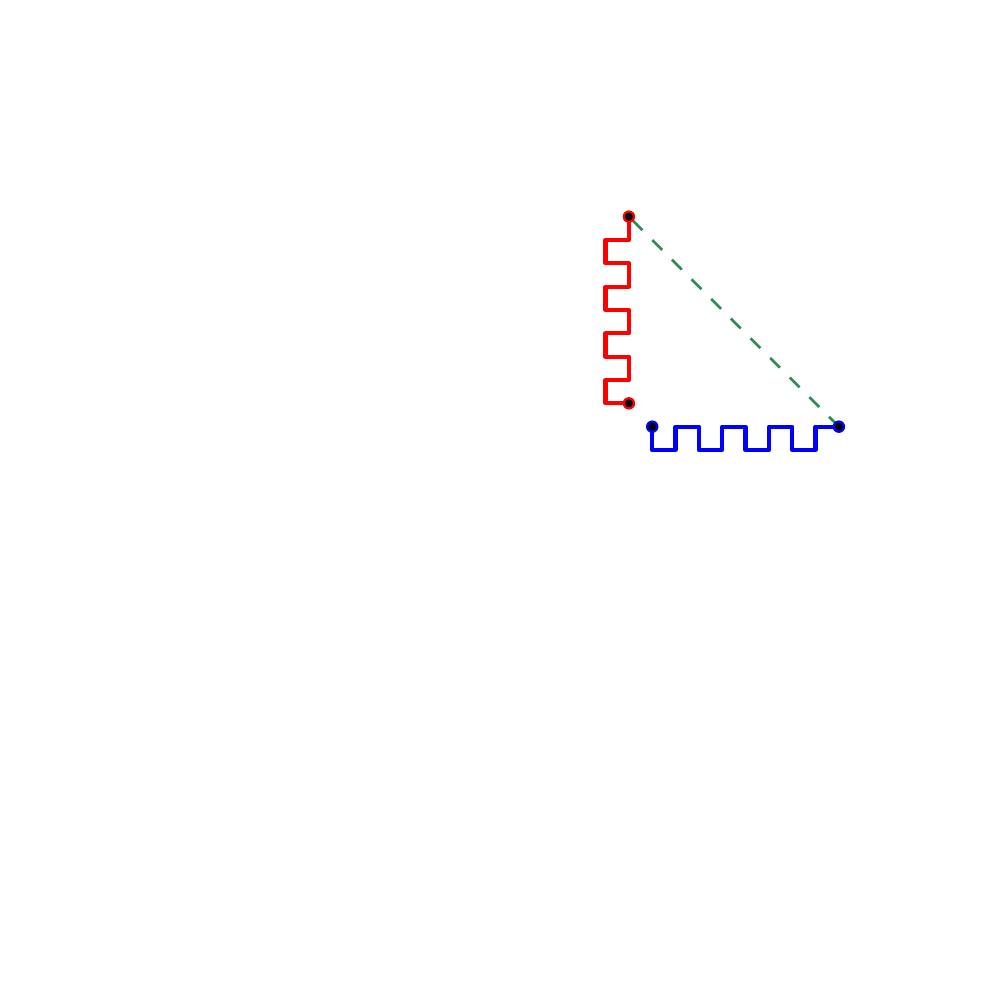}}

Interestingly, simplification is critical for the efficiency of the
above algorithm.  Indeed, consider the two nicely behaved curves
depicted on the right. The \relevant portion of the free space diagram
of these two curves, for the distance realizing the \Frechet distance,
covers a quadratic number of cells.

\parpic[l]{\includegraphics{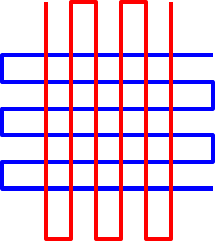}}

The use of simplification by itself is not sufficient to guarantee
that the presented algorithm is efficient.  Indeed, in might not be
possible to simplify the input curves at all without losing too much
information.  In such contrived worst case examples, the free space
diagram still has quadratic complexity due to the inherent structure
of the curves.  See the figure to the left for one such example.  In
the next section we will analyze the \resemblance using realistic
input models and prove the efficiency of the above algorithm, given
that the input is ``realistic''.

\ShowComments{ \AnneX{I reworded the paragraph above, please check.}
   \SarielX{Rewrote it again. Have a loot.}  }

\section{The \Resemblance of Families of Curves}
\seclab{resemblance}

In this section we are going to bound the \resemblance for different
realistic input models of curves. We will introduce the new class of
$c$-\emph{packed} curves, and we compare this new input model to the
previous models of \emph{$\kappa$-boundedness} and \emph{low density}.

\subsection{On \TPDF{c}{c}-packed Curves}
\seclab{c:packed:curves}

We introduce a new family of curves, $c$-packed curves, and prove that
their \resemblance $\SimpComplexity{\eps}{\curveA}{\curveB}$ is
linear, for any two curves $\curveA$ and $\curveB$ in this family.
This implies that \thmref{main} works in near linear time for
$c$-packed curves, which is one of our main results.

\subsubsection{Definition and basic properties}

\begin{defn}
    A curve $\curveA$ in $\Re^d$ is \emphi{$c$-packed} if for any
    point $\pntA$ in $\Re^d$ and any radius $r > 0$, the total length
    of $\curveA$ inside the ball $\BallX{\pntA,r}$ is at most $c r$.
\end{defn}

\begin{lemma}
    Let $\curveA$ be a curve in $\Re^d$, $\sRadius > 0$ be a
    parameter, and let $\curveAs = \simpX{\curveA, \sRadius}$ be the
    simplified curve. Then $\lenX{\curveA \cap \BallX{\pntA, \Radius +
          \sRadius}} \geq \lenX{\curveAs \cap \BallX{\pntA,\Radius}}$
    for any ball $\BallX{\pntA, \Radius}$.
    
    \lemlab{hippo}
\end{lemma}

\begin{proof}
    Let $\segA$ be a segment of $\curveAs$ that intersects
    $\BallX{\pntA, \Radius}$ and let $\segB = \segA \cap \BallX{\pntA,
       \Radius}$ be this intersection.  Let $\curveA_\segA$ be the
    portion of $\curveA$ that got simplified into $\segA$. Observe
    that $\curveA_\segA$ is a polygonal curve that lies inside a
    hippodrome of radius $\sRadius$ around $\segA$; that is,
    $\curveA_\segA \subseteq \Hippodrome_\segA = \segA \oplus
    \BallX{0, \sRadius}$, where $\oplus$ denotes the Minkowski sum of
    the two sets, see the figure on the right.
    
    \parpic[r]{\includegraphics[scale=0.9]{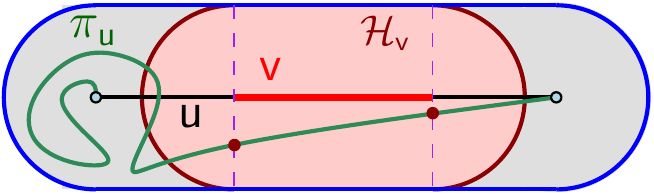}}

    In particular, erect two hyperplanes passing through the endpoints
    of $\segB$ that are orthogonal to $\segB$, and observe that
    $\curveA_\segA$ must intersect both hyperplanes.  Hence, we
    conclude that the portions of $\curveA_\segA$ in the hippodrome
    $\Hippodrome_\segB = \segB \oplus \BallX{0,\sRadius}$ are of
    length at least $\lenX{\segB}$.  Clearly, $\segB \subseteq
    \BallX{\pntA, \Radius}$ implies that $\Hippodrome_\segB \subseteq
    \BallX{\pntA, \Radius + \sRadius}$, which in turn implies that
    $\curveA_\segA \cap \Hippodrome_\segB \subseteq \BallX{\pntA,
       \Radius + \sRadius}$ and thus $\lenX{\curveA_\segA \cap
       \BallX{\pntA, \Radius + \sRadius} } \geq \lenX{\segB}$.
    
    Summing over all segments $\segB$ in $\curveAs \cap
    \BallX{\pntA,\Radius}$ implies the claim.
\end{proof}

\begin{lemma}
    Let $\curveA$ be a $c$-packed curve in $\Re^d$, $\sRadius > 0$ be
    a parameter, and let $\curveAs = \simpX{\curveA, \sRadius}$ be the
    simplified curve. Then, $\curveAs$ is a $6c$-packed curve.
    
    \lemlab{6:c:packed}
\end{lemma}

\newcommand{\segIn}{U}
\begin{proof}
    Assume, for the sake of contradiction, that $\lenX{ \curveAs \cap
       \BallX{\pntA, \Radius}} > 6c \Radius$ for some
    $\BallX{\pntA,\Radius}$ in $\Re^d$.  If $\Radius \geq \sRadius$,
    then set $\Radius'=2\Radius$ and \lemref{hippo} implies that
    $\lenX{ \curveA \cap \BallX{\pntA, \Radius'} } \geq \lenX{ \curveA
       \cap \BallX{\pntA, \Radius + \sRadius } } \geq \lenX{ \curveAs
       \cap \BallX{\pntA, \Radius}} > 6c \Radius = 3c\Radius'$, which
    contradicts that $\curveA$ is $c$-packed.
    
    If $\Radius < \sRadius$ then let $\segIn$ denote the segments of
    $\curveAs$ intersecting $\BallX{\pntA,\Radius}$ and let
    $k=|\segIn|$.  Observe that $k > 6c\Radius/2\Radius = 3c$, as any
    segment can contribute at most $2\Radius$ to the length of
    $\curveAs$ inside $\BallX{\pntA, \Radius}$. Therefore we have that
    $%
    \lenX{\curveAs \cap \BallX{\pntA, 2\sRadius}} \geq \lenX{\curveAs
       \cap \BallX{\pntA, \Radius + \sRadius}} \geq \lenX{\segIn \cap
       \BallX{\pntA, \Radius + \sRadius}} \geq k\sRadius%
    $, %
    since every segment of the simplified curve $\curveAs$ has a
    minimal length of $\sRadius$.  By \lemref{hippo}, this implies
    that $\lenX{\curveA \cap \BallX{\pntA, 3\sRadius}} \geq
    \lenX{\curveAs \cap \BallX{\pntA, 2\sRadius}} \geq k \sRadius > 3
    c \sRadius$, which is a contradiction to the $c$-packedness of
    $\curveA$.
\end{proof}

\subsubsection{Bounding the \resemblance}

\begin{lemma}
    For any two $c$-packed curves $\curveA$ and $\curveB$ in $\Re^d$,
    and $0 < \eps < 1$, we have that
    $\SimpComplexity{\eps}{\curveA}{\curveB} = O( cn /\eps)$.
    
    \lemlab{complexity:l:e:q}
\end{lemma}

\begin{proof}
    Let $\delta \geq 0$ be an arbitrary number, $\sRadius =\eps
    \delta$, $\curveAs = \simpX{\curveA,\sRadius}$ and $\curveBs =
    \simpX{\curveB,\sRadius}$
    
    We need to show that the complexity of
    $\FullFDleq{\delta}(\curveAs,\curveBs)$ is $O(c n/\eps)$.  A free
    space cell of $\FullFDleq{\delta}(\curveAs, \curveBs)$ corresponds
    to two segments $\segA \in \curveAs$ and $\segB \in \curveBs$.
    The free space in this cell is non-empty if and only if there are
    two points $\pntA \in \segA$ and $\pntB \in \segB$ such that
    $\distX{\pntA}{\pntB} \leq \delta$.  We charge this pair of points
    to the shorter of the two segments. We claim that a segment cannot
    be charged too many times.
    
    \parpic[r]{\includegraphics[scale=1]{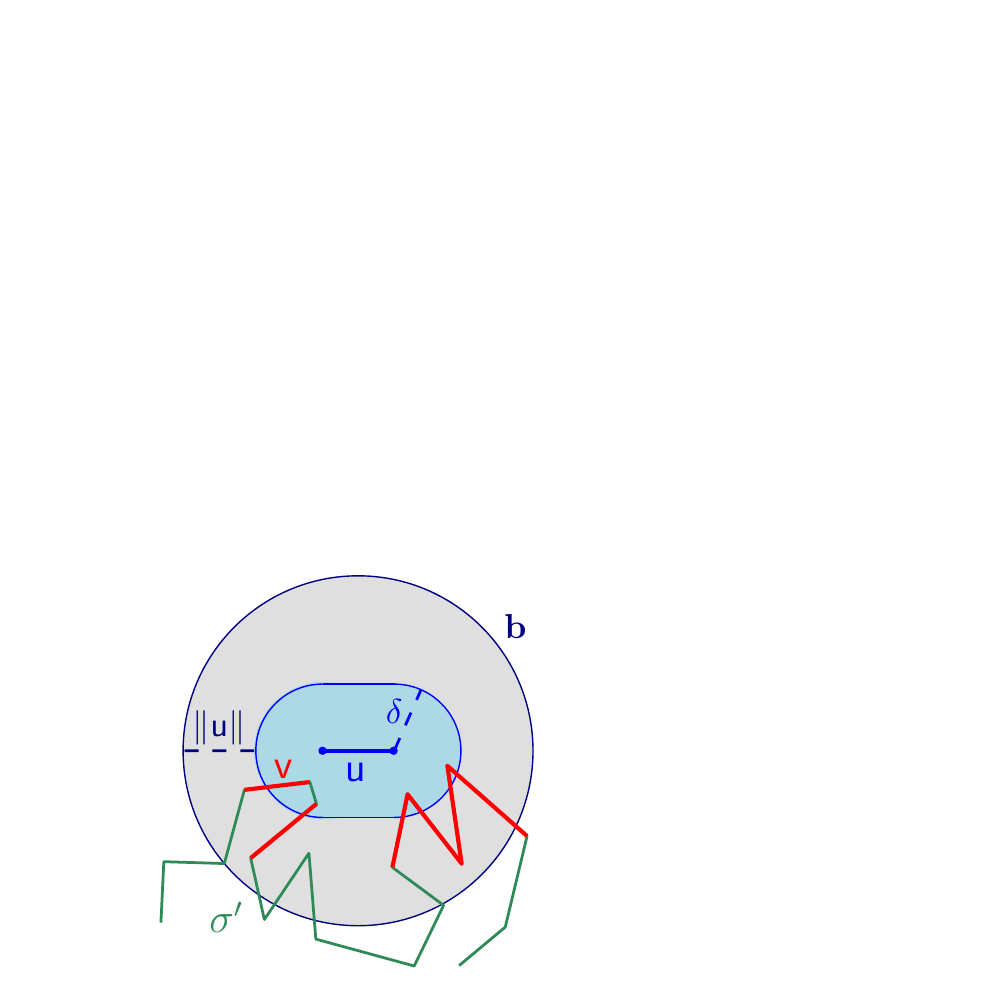}}
    Indeed, consider a segment $\segA \in \curveAs$, and consider the
    ball $\BallC$ of radius $r= (3/2)\lenX{\segA} +\delta$ centered at
    the midpoint of $\segA$, see the figure on the right.  Every
    segment $\segB \in \curveBs$ that participates in a close pair as
    above and charges $\segA$ for it, is of length at least
    $\lenX{\segA}$, and the length of $\segB \cap \BallC$ is at least
    $\lenX{\segA}$. Since $\curveBs$ is $6c$-packed, by
    \lemref{6:c:packed}, we have that the number of such charges is at
    most
    \begin{align*}
        c' = \frac{\lenX{\curveBs \cap \BallC}}{\lenX{\segA}} \leq
        \frac{6 c r}{\lenX{\segA}}%
        =%
        \frac{6c((3/2)\!\lenX{\segA} +\delta)}{\lenX{\segA}}%
        \leq 9 c + \frac{6c\delta}{\sRadius} = O\pth{\frac{c}{\eps}},
    \end{align*}
    since $\lenX{\segA} \geq \sRadius$.
    
    We conclude that there are at most $c'n$ free space cells that
    contain a point of $\FullFDleq{\delta}$. The complexity of the
    free space inside a cell is a constant, thus implying the claim.
\end{proof}

By plugging the above into \thmref{main}, we get the following result.

\begin{theorem}
    Given two polygonal $c$-packed curves $\curveA$ and $\curveB$ with
    a total of $n$ vertices in $\Re^d$, and a parameter $1 > \eps >
    0$, one can $(1+\eps)$-approximate the \Frechet distance between
    $\curveA$ and $\curveB$ in $\ds O( c n /\eps + cn \log n)$ time.
    
    \thmlab{main:c:packed}
\end{theorem}

% ------------------------------------------------------------------
% ------------------------------------------------------------------
% \newpage

\subsection{\Resemblance of Low Density Curves}
\seclab{low:density}

\begin{defn}
    A polygonal curve $\curveA$ in $\Re^d$ is
    \emphi{$\density$-low-density} if any ball $\BallX{\pntA,r}$
    intersects at most $\phi$ segments of $\curveA$ that are longer
    than $r$.
\end{defn}

% \begin{observation}
First, observe that this input model is less restrictive than the input
model which describes c-packed curves.  It can be easily seen by a
simple packing argument that a polygonal $c$-packed curve is
$\density$-low-density, for $\density=2c$.  For any ball
$\BallC=\BallX{\pntA,r}$, consider the ball with the same center that
has radius $r'=2r$. Any edge intersecting $\BallC$ that is longer than
$r$ must contribute at least $r$ to the length of the intersection of
the curve with the larger ball, which is bounded by $cr'$. There can
be at most $cr'/r=2c$ edges of this type.

% \end{observation}
\parpic[r]{\includegraphics[scale=0.4]{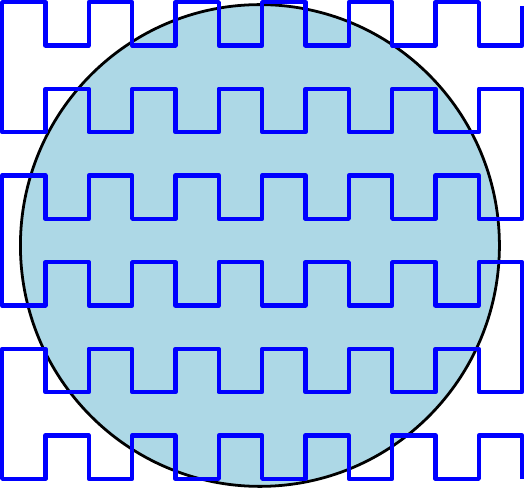}}
A curve that is low density, however, is not necessarily $c$-packed
for a small value of $c$.  Indeed, a low density curve $\curveA$ might
have an arbitrarily long intersection with a ball by having
sufficiently small segments, see the figure on the right.  However, in
this case $\curveA$ must have many vertices in the areas where its
length cannot be bounded, as we will show in the following section.

\begin{claim}
    Let $\curveA$ be a $\phi$-low density polygonal curve, and let
    $\Cube$ be a hypercube in $\Re^d$ with sidelength $\ell$. Then,
    the number of edges of length $\geq \ell$ of $\curveA$ that
    intersect $\Cube$ is bounded by $\cDim \phi$, where $\cDim =
    \ceil{\sqrt{d} / 2 }^d$.
    
    \clmlab{intersections}
\end{claim}
\begin{proof}
    Partition the cube $\Cube$ into a $D\times D \times \cdots \times
    D$ grid, for $D = \ceil{\sqrt{d} / 2 }$. Clearly, any edge that
    intersects $\Cube$ that has length $\geq \ell$ must intersect one
    of the hypercubes in this grid. A hypercube of this grid has
    diameter
    \begin{align*}
        \frac{\sqrt{d}\ell}{D}%
        \leq %
        \frac{\sqrt{d}\ell}{\sqrt{d} / 2 }%
        \leq 2 \ell,
    \end{align*}
    and is included in a ball of radius $\ell$.  Thus, a hypercube in
    this grid intersects at most $\phi$ such long edges. We conclude
    that there can be at most $\phi D^d$ long edges intersecting
    $\Cube$.
\end{proof}

\subsubsection{Low density curves can be long only if they pay for it}

\lemref{quadtree:decomposition} below testifies that the parts of a
low density curve, where its length cannot be bounded by a constant,
can be covered with hypercubes, such that each cube intersects at most
a constant number of edges and at most a constant number of other
cubes.  We use this construction in \lemref{number:of:vertices} to
relate the length of a low density curve to the diameter of the
covered area to the number of vertices.  One can verify
\lemref{quadtree:decomposition} using an easy modification of a lemma
from \cite{b-lsbsp-00}.  We provide a proof, for the sake of
completeness, in \apndref{helper}.

\begin{lemma}
    Let $\curveA$ be a $\phi$-low density curve, of which $n$ edges
    are intersecting a given hypercube $\Cube$ of $\Re^d$.  The
    hypercube $\Cube$ can be covered by a set of hypercubes
    $\CubeSet$, such that
    \begin{inparaenum}[(i)]
        \item $\bigcup \CubeSet = \Cube$,
        \item $\cardin{\CubeSet} \leq 2^{2d+1} n$,
        \item any point $\pntA \in \Cube$ is covered by at most $2^d$
        hypercubes, and
        \item each hypercube of $\CubeSet$ intersects at most
        $\cDim\phi$ edges of $\curveA$, where $\cDim$ is a constant
        that depends only on the dimension $d$.
    \end{inparaenum}
    \lemlab{quadtree:decomposition}
\end{lemma}

\begin{lemma}
    Let $\curveA$ be a $\phi$-low density curve in $\Re^d$, and let
    $\Cube$ be a cube in $\Re^d$ with side length $r$. Let $\alpha =
    \lenX{\curveA \cap \Cube}$.  There must be at least
    $\Omega((\alpha/ r)^{1+1/(d-1)})$ vertices of $\curveA$ contained
    in $3\Cube$, where $3 \Cube$ is the scaling of $\Cube$ by a factor
    of $3$ around its center.
    
    \lemlab{number:of:vertices}
\end{lemma}

\begin{proof}
    We will first give a lower bound on the number $n$ of edges
    intersecting $\Cube$ (i.e., the edges that contribute to
    $\alpha$). Then we will account for the edges that have endpoints
    outside $3 \Cube$. So, take the $n$ edges of $\curveA$ that
    intersect $\Cube$ and construct the cover of $\Cube$ resulting
    from \lemref{quadtree:decomposition} with respect to these edges.
    
    Let $\Cube_1, \ldots, \Cube_N$ denote the cubes in this cover,
    where $r_1 \leq r_2 \leq \dots \leq r_{N}$ are the side lengths of
    the cubes used by the cover,
    respectively. \lemref{quadtree:decomposition} implies that $N \leq
    2^{d+1} d n$, and, therefore, a lower bound on $N$ would provide a
    lower bound on $n$.
    
    So, the sum of the diameters of those $N$ cubes bounds the length
    of the intersection $\alpha \leq \sum_{i=1}^{N} \cDim\phi \sqrt{d}
    r_i$, since every cube in this cover can intersect at most $\cDim
    \phi$ edges of $\curveA$.  Setting $p=d$ and $q = d/(d-1)$, we
    observe that $1/p + 1/q = 1/d + (d-1)/d = 1$, and by \Holder's
    inequality\footnote{\Holder's inequality states that $\sum_{i=1}^n
       \cardin{a_i b_i} \leq \pth{\sum_{i=1}^n
          \cardin{a_i}^{q}}^{1/q}\pth{\sum_{i=1}^n
          \cardin{b_i}^{p}}^{1/p}$ if $1/p+1/q = 1$.}, we have that
    \begin{align*}
        \sum_{i=1}^{N} r_i = \sum_{i=1}^{N} r_i \cdot 1%
        \leq%
        \pth[]{\sum_{i=1}^N r_i^d }^{1/d} \pth{ \sum_{i=1}^N
           1^q}^{1/q} %
        = %
        \pth[]{\sum_{i=1}^N r_i^d }^{1/d} N^{(d-1)/d}.
    \end{align*}
    
    \lemref{quadtree:decomposition} also implies that the sum of the
    volumes of the cubes is at most $2^d \VolumeX{\Cube}$, since every
    point in $\Cube$ is covered at most $2^d$ times by this
    cover. Therefore we have that $\sum_{i=1}^{N} r_i^d = \sum_{i=1}^N
    \VolumeX{\Cube_i} \leq 2^d \VolumeX{\Cube} = 2^d r^d$. Hence
    \begin{align*}
        \alpha %
        \leq%
        \sum_{i=1}^{N} \cDim \phi \sqrt{d} r_i%
        \leq%
        \cDim\phi \sqrt{d} \pth[]{\sum_{i=1}^N r_i^d }^{1/d}
        N^{(d-1)/d}%
        \leq%
        \cDim\phi \sqrt{d} \pth{2^d r^d }^{1/d} N^{(d-1)/d}.
    \end{align*}
    This implies that $\ds \constB\pth{ {\alpha}/{r}}^{d/(d-1)} \leq
    N$, where $\constB = \pth{{ 2 \cDim \phi \sqrt{d}}
    }^{-d/(d-1)}$. Since $N \leq 2^{2d+1}  n$, we have that $\ds
    \constC\pth{ {\alpha}/{r}}^{d/(d-1)} \leq n$, where $\ds \constC =
    \frac{1}{2^{2d+1} }\pth{{ 2 \cDim\phi \sqrt{d}} }^{-d/(d-1)}$.
    
    Now, some of these $n$ edges intersecting $\Cube$ can have both
    endpoints outside $3\Cube$. Such edges are longer than the
    sidelength of $\Cube$ and by \clmref{intersections} their number
    is bounded by $\cDim\phi$.
    
    Hence, the number of vertices of $\curveA$ inside $3\Cube$ is at
    least $n - \cDim\phi \geq \constC \pth{\alpha / r}^{d/(d-1)} -
    \cDim\phi$.
\end{proof}

\begin{remark}
    One can also prove \lemref{number:of:vertices} directly, by
    building a quadtree and arguing that for a low-density curve to be
    sufficiently long, many edges in it have to be (sufficiently)
    short, thus implying the same bound. However, the current proof is
    more intuitive and cleaner.
\end{remark}

\begin{observation}
    The bound in \lemref{number:of:vertices} is tight.  For any $m >
    0$ and any $d > 0$, consider the integer grid in $\Re^{d}$ with
    coordinates in the range $1,\ldots, m$, and compute a path that
    visits all these grid points using only the grid edges of unit
    length, which is clearly possible.
    
    Now, the resulting curve is $2^d$-low density and has length
    $\alpha = m^d-1$ and its diameter is $r=\sqrt{d}m$.
    \lemref{number:of:vertices} implies that it has
    $\Omega\pth{(\alpha/ r)^{d/(d-1)}} = \Omega\pth{m^d}$
    vertices. Since this grid has $m^d$ vertices, this is tight.
    
    \obslab{first}
\end{observation}

\subsubsection{Accounting for many \Relevant Free Space Cells}

If many columns of the free space diagram of the two simplified low
density curves contain a linear number of \relevant cells, then the
curve must be ``long'' in the vicinity of the edges corresponding to
those columns, since the simplification ensures a minimal edge length.
A similar argument holds for the rows.  Therefore, using
\lemref{number:of:vertices}, we can charge the additional \relevant
cells to vertices of the original curves. This yields the following
result.

\begin{lemma}
    For any two low density curves $\curveA$ and $\curveB$ in $\Re^d$,
    and $0 < \eps < 1$, we have that
    $\SimpComplexity{\eps}{\curveA}{\curveB} = O \pth{
       \frac{n^{2(d-1)/d}}{\eps^2}}$.
    
    \lemlab{low:complexity:d}
\end{lemma}

\begin{proof}
    Let $\delta \geq 0$ be an arbitrary radius, and let $\curveAs=
    \simpX{\curveA,\sRadius}$ and $\curveBs =
    \simpX{\curveB,\sRadius}$ be their simplifications, where
    $\sRadius = \eps \delta$. Then, we need to prove that
    $\Nleq{\delta}(\curveAs, \curveBs) = O \pth{
       \frac{n^{2(d-1)/d}}{\eps^2}}$.
    
    To this end, it suffices to bound the number of vertex-edge pairs
    $(\pntA,\segA)$, where $\pntA$ is a vertex of $\curveAs$, $\segA$
    is an edge of $\curveBs$, and the distance between $\pntA$ and
    $\segA$ is at most $\delta$ (naturally, we need to apply the same
    argument to pairs with vertices in $\curveBs$ and edges in
    $\curveAs$). The total number of such pairs bounds the total
    complexity of $\FDleq{\delta} = \FDleq{\delta}(\curveAs,
    \curveBs)$.
    
    Set $M = O \pth{ \nfrac{n^{1-2/d}}{\eps^2}}$, and associate every
    vertex-edge pair $(\pntA,\segA)$ that appears in the free space
    diagram $\FDleq{\delta}$ with the vertex $\pntA$.
    
    \parpic[r]{\includegraphics[scale=0.75]{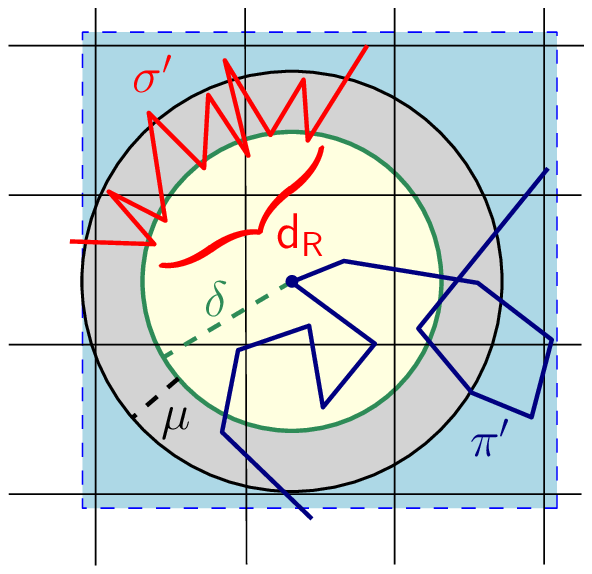}}
    
    Consider the grid $\Grid$ of side length $\delta$.  For a grid
    cell $\Cell$, consider the vertex of $\curveAs$ in $\Cell$ that is
    associated with the largest number of such vertex-edge pairs, and
    say it is being associated with $\Links_\Cell$ such vertex-edge
    pairs, and let $v_\Cell$ denote this ``popular'' vertex of
    $\curveAs$. The total number of vertex-edge pairs associated with
    vertices of $\curveAs$ inside $\Cell$ is bounded by $U_{\Cell} =
    \nVCell{\curveAs}{\Cell} \Links_\Cell$, where
    $\nVCell{\curveAs}{\Cell}$ denotes the number of vertices of
    $\curveAs$ that lie inside $\Cell$.
    
    If $\Links_\Cell \leq M$ then $U_\Cell \leq
    \nVCell{\curveAs}{\Cell} M$, and we charge $M$ units to each
    vertex of $\curveA$ inside $\Cell$.
    
    If $\Links_\Cell > M$ then the length of $\curveBs$ inside
    $\Cube/3$ is at least $\Links_\Cell \sRadius$, where $\Cube$ is a
    cube centered at $\Cell$ with side length $O(\delta)$.  Indeed,
    all the charges $\Links_\Cell$ rise from different segments of
    $\curveBs$ that are in distance at most $\delta$ from $v_\Cell$,
    and each such segment has length at least $\sRadius$.
    
    By \lemref{number:of:vertices}, we have that $\curveB$ must have
    at least $\Omega\pth{ ( \Links_\Cell \sRadius / \delta)^{d/(d-1)}
    }= \Omega\pth{ ( \Links_\Cell \eps)^{d/(d-1)} }$ vertices inside
    $\Cube$. There is some constant $c$ such that
    \begin{align*}
        &c \pth{ \eps \Links_\Cell }^{d/(d-1)}%
        \leq %
        \nVCell{\curveB}{\Cube}
        \;\implies\;%
        \Links_\Cell%
        \leq%
        \frac{1}{\eps} \pth{\frac{\nVCell{\curveB}{\Cube}}{c}
        }^{(d-1)/d} %
        \;\implies\;%
        \Links_\Cell^2 \leq \frac{1}{\eps^2}
        \pth{\frac{\nVCell{\curveB}{\Cube}}{c} }^{2-2/d}%
        \\
        &\;\implies\;%
        \Links_\Cell^2 \leq%
        \nVCell{\curveB}{\Cube} {\frac{1}{c\eps^2}
           \pth{\frac{\nVCell{\curveB}{\Cube}}{c} }^{1-2/d}}
        \leq%
        {\frac{1}{c \eps^2} \pth{\frac{n}{c} }^{1-2/d}}
        \nVCell{\curveB}{\Cube} %
        \leq%
        M \nVCell{\curveB}{\Cube},
    \end{align*}
    by picking $M$ to be sufficiently large.  In particular, if
    $\nVCell{\curveAs}{\Cell} \leq \Links_\Cell$, then $U_\Cell =
    \nVCell{\curveAs}{\Cell} \Links_\Cell\leq \Links_\Cell^2 \leq M
    \nVCell{\curveB}{\Cube}$. Hence, we charge $M$ units to each
    vertex of $\curveB$ inside the cube $\Cube$.
    
    Otherwise, $\nVCell{\curveAs}{\Cell} > \Links_\Cell > M$. But
    then, the length of $\curveAs$ inside $\Cube$ is at least
    $\nVCell{\curveAs}{\Cell} \sRadius$, and by
    \lemref{number:of:vertices}, we have that $\curveA$ must have at
    least $\Omega\pth{ ( \nVCell{\curveAs}{\Cell} \eps)^{d/(d-1)} }$
    vertices inside $\Cube$.  Arguing as above, this implies that
    $\nVCell{\curveAs}{\Cell}^2 \leq M \nVCell{\curveA}{\Cube}$. As
    such, we have that $U_\Cell = \nVCell{\curveAs}{\Cell}
    \Links_\Cell\leq \nVCell{\curveAs}{\Cell}^2 \leq M
    \nVCell{\curveA}{\Cube}$.  Again, we charge $M$ units to each
    vertex of $\curveA$ inside the cube $\Cube$.
    
    Since $\Cube$ intersects a constant number of cells of the grid,
    no vertex would get charged more than a constant number of times
    by the above scheme.  Thus, every vertex, of either curve, gets
    charged $O(M)$ units overall, and the total number of vertex-edge
    pairs present in $\FDleq{\delta}$ is $O( n M)$, as claimed.
\end{proof}

\begin{observation}
    One can extend the construction of \obsref{first} to show that
    \lemref{low:complexity:d} is close to being tight. Indeed,
    consider the grid curve of \obsref{first} in $d-1$ dimensions, for
    an integer $m$. We now lift it to $d$ dimensions by considering
    the $[1,m]^d$ cube and placing two copies of the above curve on
    two opposite faces of the cube, denoted by $f$ and $f'$. Let
    $\curveA_1$ and $\curveA_2$ denote these two copies.
    
    Next, delete the even edges from $\curveA_1$ and the odd edges
    from $\curveA_2$. Connect every vertex $v_1$ of $\curveA_1$ to its
    corresponding (copied) vertex $v_2$ in $\curveA_2$ by a path made
    out of the $m-1$ unit edges along the grid line connecting the two
    vertices. This results in a curve $\curveA$ that is similar to the
    curve constructed in \obsref{first}, but has the advantage that
    when simplified for the distance $\sRadius = m$ it results in a
    curve with $m^{d-1}$ segments of length $\geq m$ that connects
    points that lie on $f$ and on $f'$, respectively.
    
    Let $\curveB$ be a copy of $\curveA$. For a fixed $\eps >0 $, we
    can add a single segment to $\curveA$ such that the \Frechet
    distance between the resulting curves is exactly $\delta =
    m/\eps$.
    Now, these two curves have $n=2m^d+2$ vertices overall, and
    furthermore, when we simplify them for the distance $\sRadius =
    \eps \delta = m$, we end up with two curves such that every long
    edge of $\curveAs$ is going to be in distance $\leq \delta =
    m/\eps$ from a constant fraction of the edges of $\curveBs$ (this
    would be all the edges if $1/\eps > \sqrt{d}$).  Therefore the
    complexity of the \relevant free space is
    $\Omega\pth{\nVertices{\curveAs} \nVertices{\curveBs}} =
    \Omega\pth{\pth{m^{d-1}}^2} = \Omega\pth{ n^{2(d-1)/d}}$, where
    $\nVertices{\curveAs}$ denotes the number of vertices of
    $\curveAs$.  The upper bound of \lemref{low:complexity:d} is
    (only) larger by a factor of $O(1/\eps^2)$.
\end{observation}

By plugging the above into \thmref{main}, we get the following result.

\begin{theorem}
    Given two low-density curves $\curveA$ and $\curveB$ with a total
    of $n$ vertices in $\Re^d$, and a parameter $\eps > 0$, there
    exists an algorithm which $(1+\eps)$-approximates the \Frechet
    distance between $\curveA$ and $\curveB$ in $\ds O \pth{
       \frac{n^{2(d-1)/d}}{\eps^2} + n^{2(d-1)/d} \log n }$ time.
    
    \thmlab{main:low:density}
\end{theorem}

\subsection{\Resemblance of \TPDF{$\kappa$}{k}-Bounded %
   Curves}
\seclab{kappa:bounded:curves}

% \paragraph{$\kappa$-straight and $\kappa$-bounded curves.}
% \subsubsection{Previous Input Models for Curves.}

We revisit the definitions of Alt \etal \cite{akw-cdmpc-04} of
$\kappa$-bounded and $\kappa$-straight curves. Note that these
definitions describe an extremely restricted class of curves while
$c$-packed curves form a fairly general and natural class of
curves. However, it is not true that any $\kappa$-bounded curve is
$O(\kappa)$-packed. We therefore give a separate proof to bound the
\resemblance of $\kappa$-bounded curves in order to improve upon the
result in \cite{akw-cdmpc-04}.

\vspace{.5cm}
\parpic[r]{\includegraphics{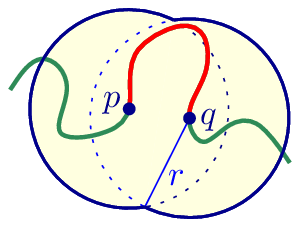}}
\vspace{-0.5cm}
\begin{defn}
    Let $\kappa\geq 1$ be a given parameter.  A curve $\curveA$ is
    \emphi{$\kappa$-straight} if for any two points $\pntA$ and
    $\pntB$ on the curve, it holds that $\lenX{\curveA[\pntA,\pntB]}
    \leq \kappa \distX{\pntA}{\pntB}$.
    % \end{defn}
    
    % \begin{defn}
    A curve $\curveA$ is a \emphi{$\kappa$-bounded} if for all $\pntA,
    \pntB \in \curveA$ it holds that the curve $\curveA[\pntA, \pntB]$
    is contained inside $\BallX{\pntA, r } \cup \BallX{\pntB, r }$,
    where $r = (\nfrac{\kappa}{2}) \distX{\pntA}{\pntB}$, see the
    figure on the right.
\end{defn}

\begin{figure}[bt]
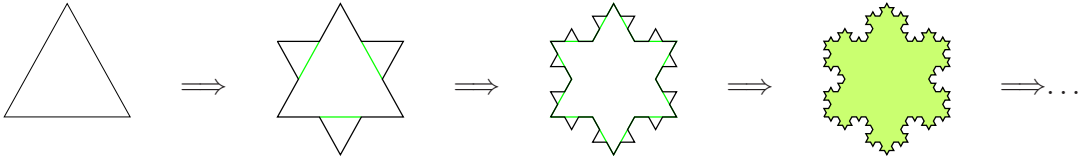

    \centerline{\begin{tabular}{cccccccc}
           \IncGraphPageExt{figs}{snowflake}{1}{width=2cm} &\RX{}
           &\IncGraphPageExt{figs}{snowflake}{2}{width=2cm} &\RX{}
           &\IncGraphPageExt{figs}{snowflake}{3}{width=2cm} &\RX{}
           &\IncGraphPageExt{figs}{snowflake}{5}{width=2cm}
           &\RX{\ldots}
       \end{tabular}}
    
    \vspace{-0.3cm}
    
    \caption{Koch's snowflake is an example of a $\kappa$-bounded
       curve that has infinite length but a finite diameter.}
    \figlab{snowflake}
\end{figure}

\begin{lemma}
    A $\kappa$-straight curve is $2\kappa$-packed.
    \lemlab{kappa:straight}
\end{lemma}
\begin{proof}
    Let $\curveA$ be a $\kappa$-straight curve in $\Re^d$, and
    consider any ball $\BallX{\pntA,r}$ that intersects it. Let
    $\pntB$ and $\pntC$ be the first and last points, respectively,
    along $\curveA$ that are in $\BallX{\pntA,r}$. Clearly,
    $\distX{\pntB}{\pntC} \leq 2r$, and by the $\kappa$-straightness $
    \lenX{\curveA \cap \BallX{\pntA,r}} \leq
    \lenX{\curveA[\pntB,\pntC]} \leq \kappa \distX{\pntB}{\pntC} \leq
    2 \kappa r$.
\end{proof}

\begin{remark}
    It is easy to verify that a $\kappa$-straight curve is also
    $\kappa$-bounded. However, $\kappa$-bounded curves,
    counterintuitively, can have infinite length even when contained
    inside a finite domain. An example of this is \emphi{Koch's
       snowflake}, which is a fractal curve depicted in
    \figref{snowflake}.
    
    To see, intuitively, why Koch's snowflake is $\kappa$-bounded, let
    $\curveA_i$ be the $i$\th polygonal curve generated by this
    process. There is a natural mapping between any point of
    $\curveA_i$ and $\curveA_{i+1}$, for all $i$. In particular,
    consider two points $\pntA$ and $\pntB$ on the final curve
    $\curveA^*$, and consider the two sequences of points
    $\pntA_i,\pntB_i \in \curveA_i$, where $\pntA_{i+1} \in
    \curveA_{i+1}$ (resp. $\pntB_{i+1} \in \curveA_{i+1}$) is the
    natural image of $\pntA_{i}$ (resp. $\pntB_i$), %
    $\lim_{i \rightarrow \infty} \pntA_i =\pntA$, and $\lim_{i
       \rightarrow \infty} \pntB_i =\pntB$.
    
    Now, assume that $r = \distX{\pntA}{\pntB}$. Observe that, for all
    $i$, the polygonal curve $\curveA_i$ is made out of segments that
    are all of the same length. In particular, consider the first
    index $k$, such that this segment length of $\curveA_k$ is of
    length $\leq r/20$. It is easy to argue that
    $\distX{\pntA_{k}}{\pntA} \leq r/5$ and $\distX{\pntB_{k}}{\pntB}
    \leq r/5$. In fact, one can argue that no point of $\curveA_k$
    moves more than a distance larger than $r/5$ to its final location
    on $\curveA^*$.
    
    Now, a tedious argument shows that there are $O(1)$ segments of
    $\curveA_k$ separating $\pntA_k$ from $\pntB_{k}$. Therefore this
    portion of the curve $\curveA_k$ is covered by a disk of radius
    $O(r)$, and the corresponding portion of the final curve between
    $\pntA$ and $\pntB$ is also covered by a disk of radius
    $O(r)$. This implies that Koch's snowflake is $\kappa$-bounded.
    
    A formal proof of this fact is considerably more tedious and is
    omitted.
    
    \remlab{formal:is:tedious}
\end{remark}

\begin{lemma}
    Let $\curveA$ be a $\kappa$-bounded polygonal curve in $\Re^d$,
    and let $\sRadius \leq \delta$ be parameters.  Let $\curveAs =
    \simpX{\curveA, \sRadius}$.  Then the number of segments of
    $\curveAs$ intersecting $\BallX{\pntC,\delta}$ is bounded by
    $O\pth{ \kappa^d (1 + \delta/\sRadius)^d}$, for any $\pntC \in
    \Re^d$.
    
    \lemlab{kappa:l:e:q}
\end{lemma}
\begin{proof}
    For $\curveA = \segA_1 \segA_2 \ldots \segA_k$, let $ Y_O = \brc{
       \segA_1, \segA_3, \ldots }$ and $Y_E = \brc{ \segA_2, \segA_4,
       \ldots }$ be the sets of odd and even segments of $\curveAs$,
    respectively.
    
    Let $X_O \subseteq Y_O$ be the set of odd segments of $\curveAs$
    intersecting $\BallX{\pntC,\delta}$. For all $i$, pick an
    arbitrary point $\pntA_i$ on the $i$\th segment of $X_O$ that lies
    inside $\BallX{\pntA, \delta}$.  Next, pick an original point
    $\pntB_i$ of $\curveA$ in distance at most $\sRadius$ from
    $\pntA_i$, for $i=1, \ldots, M = \cardin{X_O}$.  Observe, that for
    all $i$ we have that $\distX{\pntC}{\pntB_i} \leq \delta +
    \sRadius$.  Furthermore, between any two distinct points $\pntA_i$
    and $\pntA_j$ on the simplified curve $\curveAs$ there must lie an
    even segment of $Y_E$ in between them along the curve, and the
    length of this segment is at least $\sRadius$ (because the
    simplification algorithm generates segments of length at least
    $\sRadius$). Also, the endpoints of this even segment lie on the
    original curve $\curveA$.
    
    We claim that no two points of $\PntSetA = \brc{ \pntB_1, \ldots,
       \pntB_M}$ can be too close to each other; that is, there are no
    two points $\pntB', \pntB'' \in \PntSetA$, such that $r =
    \distX{\pntB'}{\pntB''} \leq \sRadius/ (4\kappa)$.  Indeed, assume
    for the sake of contradiction, that there are two such
    points. Then, by the above, the portion of $\curveA$ connecting
    them contains two points $\pntD', \pntD''$ that are at least
    $\sRadius$ apart.  Observe that $\curveA[\pntD', \pntD'']
    \subseteq X = \BallX{\pntB',(\kappa/2)
       r}\cup\BallX{\pntB'',(\kappa/2) r}$.  However, the maximum
    distance between two points that are included inside $X$ is
    bounded by its diameter. We have that
    \begin{align*}
        \sRadius \leq \distX{\pntD}{\pntD'} \leq \diameterX{X} =
        2(\kappa/2)r + \distX{\pntB'}{\pntB''}\leq \frac{\sRadius}{4}
        + \frac{\sRadius}{4\kappa}%
        \leq%
        \frac{\sRadius}{2},
    \end{align*}
    since $\kappa>1$. A contradiction.
    
    However, all the points of $\PntSetA$ lie inside a ball of radius
    $\delta+\sRadius$ centered at $\pntC$. Now, placing a ball of
    radius $\sRadius' = \sRadius/ (8\kappa)$ around each point of
    $\PntSetA$, results in a set of interior disjoint balls.  This
    implies, by a standard packing argument, that the number of points
    of $\PntSetA$ is bounded by $
    \VolumeX{\BallX{\pntC,\delta+\sRadius}} / \VolumeX{\BallX{\cdot,
          \sRadius'}} = O\pth{(\delta+\sRadius)^d/
       (\sRadius/\kappa)^d} = \pth{ (1+\delta/\sRadius)^d \kappa^d }$.
    
    This bounds the number of odd segments of $\curveAs$ intersecting
    the ball $\BallX{\pntC,\delta}$, and a similar argument holds for
    the even segments intersecting this ball.
\end{proof}

\begin{lemma}
    For any two $\kappa$-bounded polygonal curves in $\Re^d$ $\curveA$
    and $\curveB$, $0 < \eps < 1$, we have
    $\SimpComplexity{\eps}{\curveA}{\curveB} = O\pth{ (\kappa/\eps)^d
       n}$.

    \lemlab{complexity:l:e:q:kappa}
\end{lemma}

\begin{proof}
    Let $\delta \geq 0$ be an arbitrary radius, and set $\sRadius =
    \eps \delta$.  Let $\curveAs = \simpX{\curveA,\sRadius}$ and
    $\curveBs = \simpX{\curveB,\sRadius}$. We need to show that the
    complexity of the \relevant free space
    $\FDleq{\delta}(\curveAs,\curveBs)$ is $O\pth{\kappa^d(
       1+\delta/\sRadius)^d n} = O\pth{ (\kappa/\eps)^d n}$ .
    
    The boundary of a \relevant cell in the free space diagram has a
    non-empty intersection with $\FullFDleq{\delta}(\curveAs,
    \curveBs)$. Otherwise its interior could not be reached by a
    monotone path from $(0,0)$. Therefore, using an argument similar
    to the proof of \lemref{complexity:l:e:q}, \lemref{kappa:l:e:q}
    implies the desired bound.
\end{proof}

By plugging the above into \thmref{main}, we get the following result.

\begin{theorem}
    Given two $\kappa$-bounded polygonal curves $\curveA$ and
    $\curveB$ with a total of $n$ vertices in $\Re^d$, and a parameter
    $1 > \eps > 0$, there exists an algorithm which
    $(1+\eps)$-approximates the \Frechet distance between $\curveA$
    and $\curveB$ in $\ds O\pth{ (\kappa /\eps)^d n + \kappa^d n \log
       n}$ time.
    
    \thmlab{main:kappa:bounded}
\end{theorem}

\section{Extension to Closed \TPDF{$c$}{c}-packed Curves}
\seclab{closed:curves}

The \Frechet distance for closed curves is defined as in
\secref{prelims} with the altered condition that the
reparameterizations $f$ and $g$ are orientation-preserving
homeomorphisms on the one-dimensional sphere.  Computing the \Frechet
distance for closed curves is more difficult, as the constraint that
the endpoints of the curves have to be matched to each other is
dropped in this case and therefore the set of reparameterizations one
has to consider is larger.

\begin{observation}
    The decision problem for closed curves can be reduced to the
    previously considered case of open curves. Given two closed
    $c$-packed curves $\curveA$ and $\curveB$ and a parameter
    $\delta$.  Pick a vertex $\pntA$ of the curve $\curveA$, and
    assume that we know a point $\pntB$ on $\curveB$ that is being
    matched to $\pntA$ by a pair of reparameterizations of $\curveA$
    and $\curveB$ of width at most $\delta$. Clearly, if we break
    $\curveA$ open at $\pntA$, and $\curveB$ at $\pntB$, we retrieve
    two open curves $\subcurveA$ and $\subcurveB$, and we can use the
    previous method to decide if $\distFr{\subcurveA}{\subcurveB} \leq
    \delta$.  Hence we only need to generate a suitable set of
    candidates for $\pntB$ to determine if the \Frechet distance
    between $\curveA$ and $\curveB$ is at most $\delta$ within a
    certain approximation error.
\end{observation}

\begin{lemma}
    Let $\curveA$ be a closed $c$-packed polygonal curve in $\Re^d$,
    and let $\sRadius \leq \delta$ be parameters.  Let $\curveAs =
    \simpX{\curveA, \sRadius}$.  Then the number of edges of
    $\curveAs$ intersecting $\BallX{\pntA,\delta}$ is bounded by $O(c
    \delta/\sRadius)$, for any $\pntA \in \Re^d$.
    
    \lemlab{closed:l:e:q}
\end{lemma}

\begin{proof}
    Consider the ball $\BallC = \BallX{\pntA, r }$ of radius $r=
    \sRadius +\delta$.  Any edge $\segA$ of $\curveAs$ that intersects
    $\BallX{\pntA, \delta}$ has to contribute at least $\sRadius$ to
    the length of the intersection with $\BallC$, as the
    simplification guarantees that every edge of $\curveAs$ is of
    length at least $\sRadius$. Since $\curveAs$ is $6c$-packed, by
    \lemref{6:c:packed}, we have that $\lenX{\BallC \cap \curveAs}
    \leq 6c r$, and the number of intersections of $\curveAs$ with
    $\BallX{\pntA,\delta}$ is $N \leq \lenX{\BallC \cap
       \curveAs}/\sRadius \leq 6cr / \sRadius = 6
    c(\sRadius+\delta)/\sRadius = O(c +c \delta/\sRadius) $, which
    implies the claim.
\end{proof}

\begin{lemma}
    Given two closed $c$-packed polygonal curves $\curveA$ and
    $\curveB$ with a total number of $n$ vertices and parameters
    $\delta$ and $1 > \eps > 0$.  Let $\curveAs = \simpX{\curveA,
       \sRadius}$ and $\curveBs = \simpX{\curveB, \sRadius}$ denote
    the curves simplified with $\sRadius \leq \eps\delta$ and let
    $\pntA$ be a vertex of $\curveAs$.  We can compute a set of points
    $\KSet \subseteq \curveBs$ of size $O(c/\eps)$, in $O(n + c/\eps)$
    time, such that if $\distFr{\curveAs}{\curveBs} \leq \delta$ then
    there exists a pair of reparameterizations of width at most
    $(1+\eps)\delta$ that matches $\pntA$ to an element of $\KSet$.
    
    \lemlab{closed:k:set}
\end{lemma}

\begin{proof}
    We walk along the curve $\curveBs$ starting from an arbitrary
    point.  If the starting point is in distance $\delta$ from
    $\pntA$, then we add it to the candidate set $\KSet$. As we follow
    along the curve we create a candidate if we %
    \smallskip%
    \begin{compactenum}[~~~(a)]
        \item (re-)enter the ball $\BallX{\pntA, \delta}$, or
        \item have traveled a distance $\eps\delta$ along $\curveBs$
        since the last creation of a candidate, unless we have exited
        the ball $\BallX{\pntA,\delta}$ in the meantime.
    \end{compactenum}
    \smallskip %
    Clearly this takes $O\pth{n + \cardin{ \KSet } }$ time.

    The number of events of type (a) is bounded (up to a factor of 2)
    by the number of intersections of $\curveBs$ with the sphere
    $\SphereX{\pntA,\delta}$, and by \lemref{closed:l:e:q}, this
    number is bounded by $O(c\delta/\sRadius) = O(c/\eps)$.  By
    \lemref{6:c:packed} the simplified curve $\curveBs$ is $6c$-packed
    and therefore the length of its intersection with
    $\BallX{\pntA,\delta}$ is at most $6c\delta$. This implies that we
    can have at most $O(6c\delta/\sRadius) = O(6c/\eps)$ candidates
    that were created at events of type (b).

    \parpic[r]{\includegraphics{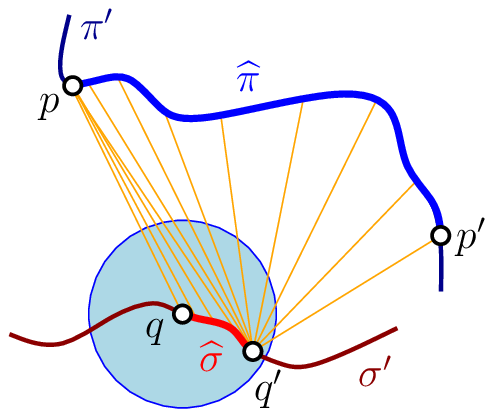}}
    
    Consider reparameterizations of $\curveAs$ and $\curveBs$ of width
    at most $\delta$. Next, consider a point $\pntB \in \curveBs$ that
    is matched to $\pntA \in \curveAs$ by these
    reparameterizations. Observe that $\pntB \in \BallX{\pntA,\delta}$
    and there exists, by construction, a point $\pntB' \in \KSet$ such
    that $\distX{\pntB}{\pntB'} \leq \eps \delta$.  Let $\pntA'$ be a
    point on $\curveAs$ that is matched to $\pntB'$ by the given
    reparameterizations.

    % We argue that a pair of reparameterizations of $\curveAs$ and
    % $\curveBs$ of width at most $\delta$ can always be altered such
    % that it matches $\pntA$ to a point in $\KSet$ while increasing
    % the width by at most a factor of $1 + \eps$.
    
    % Any by such a pair of reparameterizations of $f$ and $g$ has to
    % be contained in $\BallX{\pntA, \delta}$. As such, $\pntB$ is in
    % distance $\eps \delta$ of some point $\pntB'$ in our candidate
    % set $\KSet$.
    
    We match the curve segment $\subcurveB$ between $\pntB$ and
    $\pntB'$ to $\pntA$ and the curve segment $\subcurveA$ between
    $\pntA$ and $\pntA'$ to $\pntB$, see the figure to the right.
    Clearly this preserves the monotonicity of the matching. By the
    triangle inequality, any point on $\subcurveB$ has distance at
    most $(1+\eps)\delta$ to $\pntA$.  Similarly, for any point on
    $\subcurveA$ there is a point on $\subcurveB$ that is in distance
    $\delta$, therefore $\pntB'$ is in distance $(1+\eps)\delta$ of
    $\subcurveA$.
    
    We conclude that the \Frechet distance between $\curveAs$ and
    $\curveBs$ is at most $(1+\eps)\delta$ when restricted to
    reparameterizations matching $\pntA$ to $\pntB'$.
\end{proof}

One can adapt \lemref{decider} to the closed curves case, by
considering the $O(cn/\eps)$ open curves that result from breaking
$\curveBs$ at any point of $\KSet$. The details of the adaption are
straightforward, and we only state the result.

\begin{lemma}
    Given two closed polygonal $c$-packed curves $\curveA$ and
    $\curveB$ with a total of $n$ vertices in $\Re^d$, and parameters
    $\delta$ and $1 > \eps > 0$.  Then, there exists an algorithm
    which, in $\ds O\pth{ (c/\eps)^2n}$ time, correctly outputs one of
    the following:
    \begin{compactenum}[(A)]
        \item If $\distFr{\curveA}{\curveB} \leq \delta$ then the
        algorithm outputs ``$\leq (1+\eps)\delta$''.
        
        \item If $\distFr{\curveA}{\curveB} > (1+\eps)\delta$ then the
        algorithm outputs ``$\distFr{\curveA}{\curveB} > \delta$''.
        
        \item If $\distFr{\curveA}{\curveB} \in \pbrc{\delta,
           (1+\eps)\delta}$ then the algorithm outputs either of the
        above outcomes.
    \end{compactenum}
    
    \lemlab{decider:closed}
\end{lemma}

Plugging \lemref{decider:closed} into the algorithm of \thmref{main},
we get the following result.
\begin{theorem}
    Given two closed polygonal $c$-packed curves $\curveA$ and
    $\curveB$ with a total of $n$ vertices in $\Re^d$, and a parameter
    $1 > \eps > 0$, one can $(1+\eps)$-approximate the \Frechet
    distance between $\curveA$ and $\curveB$ in $\ds O\pth{ c^2 n
       \pth{ \eps^{-2} + \log n}}$ time.
    
    \thmlab{main:closed}
\end{theorem}

% ------------------------------------------------------------------
% ------------------------------------------------------------------

\section{Conclusions}
\seclab{conclusions}

We presented a new approximation algorithm for the \Frechet distance
for polygonal curves in any fixed dimension.  The new algorithm is
surprisingly simple and should be practical. Furthermore it works for
any kind of polygonal curves. Since the algorithm simplifies the
curves to the ``right'' resolution during the execution, we expect the
algorithm to be fast in practice.  The algorithm's analysis relies on
the concept of the \resemblance of curves, which tries to capture the
complexity of the free space diagram when simplification is being
used.

Next, we introduced the $c$-packed family of curves.  While not all
curves are $c$-packed, it seems that most real life curves are
$c$-packed.  The family of $c$-packed curves is closed under
simplification, and the property of a curve being $c$-packed is
independent of the ambient dimension of the space containing the
curve.  We expect this concept to be used to analyze other algorithms
in the future.

In particular, the \resemblance of $c$-packed curves is linear.  We
gave bounds for the \resemblance for several other types of curves,
from low density curves to $\kappa$-bounded curves.  Finally, we also
showed that the algorithm can be modified to handle closed curves
efficiently.

% , whether they are $c$-packed, low density, or not

\paragraph{Lower bound.}
Our solution to the decision problem ``beats'' the lower bound of
$\Omega(n \log n)$ \cite{bbkrw-wtd-07}, by a factor of $\log n$ (see
\lemref{decider}).  Since our decision procedure is approximated this
is not too surprising. However, it is enlightening to consider where
this proof breaks for our settings.  Indeed, Buchin \etal
\cite{bbkrw-wtd-07} generate two curves such that the \Frechet
distance might be realized by one vertex of one curve matching the
whole other curve.  On the other hand, in our case, the input model
coupled with simplification, guarantees that the number of segments
matching a single vertex is only a constant.

% \paragraph{Open problems.}
% There are many interesting questions for further research -- from
% extending this work to maps \cite{aerw-mpm-03} to whether or not one
% can come up with a similar definition for matching realistic
% surfaces.  As of now, not much is known on the efficient
% computability of the \Frechet distance between surfaces
% \cite{g-cms-99, ab-csbs-10}.  Finally, it is clear that a variant of
% our algorithm also works for non-polygonal curves, and in some cases
% even if they do not have finite complexity. It would be interesting
% to further pursue this direction.
% 

% ------------------------------------------------------------------
% ------------------------------------------------------------------

\subsubsection*{Acknowledgments.}

The authors thank Mark d{}e Berg and Marc van Kreveld for insightful
discussions on the problems studied in this paper and related
problems.  The authors would also like to thank the anonymous referees
for their insightful comments.

This research was initiated during a workshop supported by the
Netherlands Organization for Scientific Research (NWO) under
BRICKS/FOCUS grant number 642.065.503.

% -------------------------------------------------------------------------

\bibliographystyle{alpha}
%
%\bibliography{shortcuts,geometry}
\bibliography{frechet}

\appendix

\section{Fatness implies \TPDF{$c$}{c}-packedness}
\apndlab{a:b:covered:c:pack}

We show that the boundary of an $(\alpha,\beta)$-covered shape is
$c$-packed even if the shape does not have a finite descriptive
complexity. A somewhat similar result (which however is too weak to
prove this result) is the packing lemma of d{}e Berg
\cite{d-ibucfo-08} that shows that the boundary of the union of
$\gamma$-fat shapes has low density.  This implies that a connected
component of this boundary has low density.

As mentioned in the introduction, since Koch's snowflake is
$\gamma$-fat, if the finiteness requirement is removed, it follows
that the boundary of $\gamma$-fat shapes with unbounded descriptive
complexity are not $c$-packed, for any finite $c$.

\begin{defn}
    A bounded simply connected region $P$ in the plane is
    \emphi{$(\alpha,\beta)$-covered} if for each point $\pntA
    \in \partial P$, there exists a triangle $T_\pntA$, called a
    \emphi{good triangle} of $\pntA$, such that:
    \begin{inparaenum}[(i)]
        \item $\pntA$ is a vertex of $T_\pntA$,
        \item $T_\pntA \subseteq P$,
        \item all the angles of $T_\pntA$ are at least $\alpha$, and
        \item the length of all the edges of $T_\pntA$ is at least
        $\beta \diameterX{ P}$.
    \end{inparaenum}
\end{defn}

Note, that our definition is different from the standard definition of
$(\alpha,\beta)$-covered shapes, since we do not require that the
region $P$ has a finite descriptive complexity.

\begin{lemma}
    Let $S$ be a set of segments contained inside a disk with radius
    $r$, such that for any point $\pntA$ lying on a segment of $S$,
    there is an infinite cone $\Cone$ of angle at least $\alpha \leq
    \pi $ with an apex at $\pntA$, such that the intersection of the
    interior of $\Cone$ with $S$ is empty. Then, $\lenX{S} \leq 10 \pi
    r /(\alpha \sin(\alpha/4))$.
    
    \lemlab{trivial}
\end{lemma}
\begin{proof}
    Let $\Family$ be a family of $\ceil{ 2\pi/ (\alpha/2)}$ cones,
    centered at the origin, such that they cover all directions, and
    each cone has angle $\alpha/2$. Clearly, for any point $\pntA$
    lying on a segment of $S$, there must be a cone $\Cone \in
    \Family$, such that the interior of $\pntA + \Cone$ does not
    intersect $S$. We will say that $\pntA$ is \emphi{exposed} by
    $\Cone$.
    
    \parpic[r]{\includegraphics[scale=0.75]{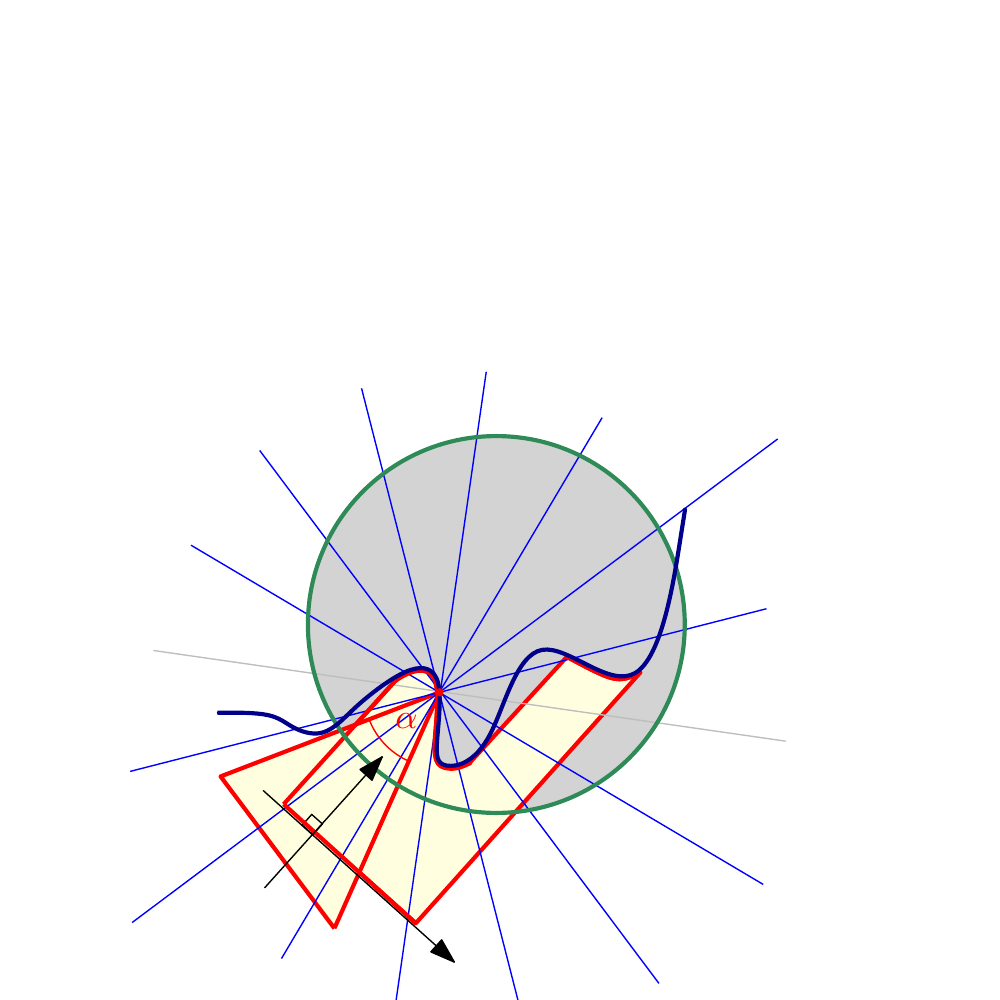}}
    So, fix such a cone $\Cone \in \Family$ and consider the direction
    $\vec{v}$ that splits the angle of $\Cone$ into two. Rotate the
    plane such that $\vec{v}$ is the direction of the negative $y$
    axis, and observe that any point of $S$ that is exposed by (the
    rotated) $\Cone$ lies on the lower envelope of the segments of
    $S$. Furthermore, the segment $\segA \in S$ that contains this
    point must have an angle in the range $(-\pi/2 + \alpha/4, \pi/2 -
    \alpha/4)$ with the positive direction of the $x$-axis (we assume
    $\segA$ is oriented from left to right).
    
    Now, since the projection of $S$ on the $x$-axis has length at
    most $2r$, it follows that the total length of the segments
    exposed by $\Cone$ is at most $2r /\sin(\alpha/4)$.
    
    Hence, the total length of segments of $S$ is bounded by
    \begin{align*}
        \cardin{\Family} \left(\frac{2r}{\sin(\alpha/4)} \right) =
        \left(\frac{4\pi}{\alpha} +
            1\right)\left(\frac{2r}{\sin(\alpha/4)}\right) \leq
        \frac{10 \pi r }{\alpha \sin(\alpha/4)}.
    \end{align*}
    \aftermathA
\end{proof}

\begin{lemma}
    If $P$ is an $(\alpha,\beta)$-covered polygon in the plane then it
    is $c$-packed, for
    \begin{align*}
        \displaystyle c = O\pth{ \frac{1}{\alpha \beta \sin(\alpha/4)
              \tan(\alpha)} }.
    \end{align*}
    \lemlab{fatness}
\end{lemma}
\begin{proof}
    Let $S = \partial{P}$, and consider any disk $D$ of radius $r$ in
    the plane. Observe that the height of a good triangle is at least
    $\rho = (s/2) \tan(\alpha)$, for $s=\beta \diameterX{P}$, and this
    also bounds the distance of any vertex of a good triangle to its
    facing edge.  If $r \leq \rho/2$, then any good triangle for a
    point of $S$ behaves like a cone as far as $S \cap D$ is
    concerned, and \lemref{trivial} implies that $\lenX{S \cap D} \leq
    10 \pi r /(\alpha \sin(\alpha/4))$ as desired.
    
    If $r \leq \diameterX{P}$, then cover $D$ by $m = (2\sqrt{2}
    r/\rho +1)^2$ disks of radius $\rho/2$. Clearly, for each such
    disk, the total length of segments of $S$ inside it, by
    \lemref{trivial}, is at most $5 \pi \rho /(\alpha
    \sin(\alpha/4))$. Therefore the total length of $S$ inside $D$ is
    \begin{align*}
        \frac{5 \pi \rho}{(\alpha \sin(\alpha/4))}
        \pth{\frac{2\sqrt{2} r}{\rho} +1}^2 \leq \frac{ 160
           \pi}{\alpha \sin(\alpha/4)} \cdot \frac{r}{\rho} \cdot r
        \leq \frac{320 \pi}{\alpha \beta \sin(\alpha/4) \tan(\alpha)}
        r.
    \end{align*}
    
    Observe that the total length of $\partial P$ is bounded by the
    above bound, by taking $D$ to be a disk of radius $r =
    \diameterX{P}$ centered at some point of $P$. Therefore the claim
    trivially holds in the case $r \geq \diameterX{P}$.
\end{proof}

\section{A proof of \lemref{quadtree:decomposition}}
\apndlab{helper}

\begin{lemma}
    Let $\PntSet$ be a set of $n$ points in $\Re^d$, contained inside
    a hypercube $\Cube$. Then one can cover $\Cube$ by a set of cubes
    $\CubeSet$, such that the following properties hold.
    \begin{compactenum}[(A)]
        \item $\bigcup \CubeSet = \Cube$.
        \item $\cardin{\CubeSet} \leq 2^{d+1} dn$.
        \item Each $\pntA \in \Cube$ is covered by at most $2^d$
        cubes.
        \item Each cube contains at most one point from $\PntSet$.
    \end{compactenum}
    \lemlab{quadtree:points}
\end{lemma}

\begin{proof}
    We can use the following algorithm to construct a $d$-dimensional
    reduced quadtree, of which the set of cubes corresponding to the
    leaf nodes satisfies the requirements for $\CubeSet$.
    
    \parpic[r]{\includegraphics[scale=0.75]{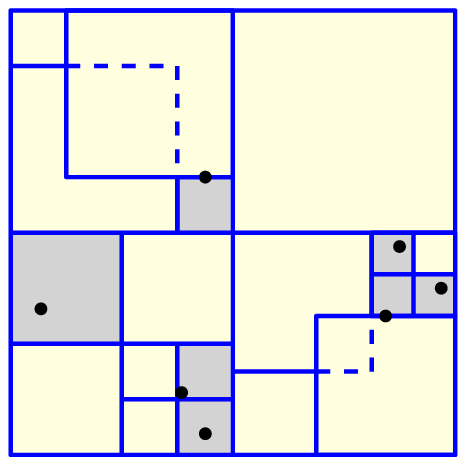}}
    Take $\Cube$ as the root node.  Split the current node recursively
    into $2^d$ subcubes, until there is only one point left in the
    current node, while abiding to the following rule.
    
    In each step, either (A) do a proper quadtree split if at least
    two of the immediately resulting subcubes contain a point of
    $\PntSet$, or (B) perform a \emphi{reduced split} otherwise, such
    that all points are contained in exactly one minimal subcube.  A
    reduced split is formed by allowing the cubes the cube to overlap,
    by shrinking one of the $2^d$ subcubes containing the points, and
    enlarging all the others.  Such a reduced split is depicted in the
    figure to the right. However, we can assure that only those
    subcubes will overlap that do not contain any point of $\PntSet$
    and will therefore not be split further. Clearly each point in the
    covered area is covered by at most $2^d$ leaf nodes.
    
    A split of type (A) separates the set of points into at least two
    non-empty subsets. A split of type (B) results in a point on the
    splitting plane. Both events can happen at most $dn$ times and
    produce each at most $2^d$ extra nodes. Therefore the size of
    $\CubeSet$ is bounded by $2^{d+1} dn$.
\end{proof}

\begin{fake_proof}
    \emph{Proof of \lemref{quadtree:decomposition}:} This follows
    directly from \lemref{quadtree:points}.  Indeed, for every edge of
    $\curveA$ add the corners of the axis parallel cube containing it
    to a set of points $\PntSet$.  Next, consider the respective
    quadtree construction of \lemref{quadtree:points} for
    $\PntSet\subseteq \Cube$.  The cover uses at most $ 2^{d+1} m $
    boxes, where $m \leq 2^d n = \cardin{\PntSet}$.
    
    Consider a cube $\Cube'$ in the resulting decomposition of
    $\Cube$, and an edge $\segA$ of $\curveA$ that intersects it. If
    the length of $\segA$ is shorter than the sidelength of $\Cube'$,
    then one of the corners of the bounding cube of $\segA$ must be in
    $\Cube'$, and $\Cube'$ cannot be a leaf of the quadtree. This
    implies that $\Cube'$ can be intersected only by edges that are at
    least as long as its sidelength.
    
    By the low density property of $\curveA$ and by
    \clmref{intersections}, $\Cube'$ can intersect at most $\cDim
    \phi$ edges of $\curveA$, which implies the lemma.
\end{fake_proof}

\end{document}